\documentclass[a4paper,11pt]{article}
\usepackage[left=3cm, right=3cm, top=3cm, bottom=3cm]{geometry}
\usepackage{amsmath,amsfonts,amssymb,amsthm}
\usepackage{mathrsfs}
\usepackage{color}
\usepackage{cancel}
\usepackage{graphicx}
\usepackage{caption}
\usepackage{subcaption}
\usepackage{gensymb}
\usepackage{scalerel}
\usepackage{bm}
\usepackage{systeme}
\usepackage{comment}
\usepackage{dcolumn,booktabs}
\usepackage{empheq}
\usepackage{hyperref}
\usepackage{makecell}
\usepackage{float}
\usepackage{mathtools}

\def\N{\mathbb{N}}
\def\R{\mathbb{R}}
\def\uvp{\mathfrak{u}} 
\def\coll{\mathscr{C}}
\def\dnod{d_{\rm nod}}
\def\dmin{d_{\rm min}}
\def\deltanod{\delta_{\rm nod}}
\def\rpiu{r_{\scaleto{+}{5pt}}}
\def\rmeno{r_{\scaleto{-}{5pt}}}
\def\rpiup{r_{\scaleto{+}{5pt}}'}
\def\rmenop{r_{\scaleto{-}{5pt}}'}

\def\uppintom{u^\omega_{\rm int}}
\def\uppextom{u^\omega_{\rm ext}}
\def\upplinkom{u^\omega_{\rm link}}
\def\omegamutuno{\omega_1^{\scaleto{M}{3pt}}}
\def\omegamutdue{\omega_2^{\scaleto{M}{3pt}}}
\def\C{\mathbb{C}}

\newtheorem{proposition}{\bf Proposition}
\newtheorem{remark}{\bf Remark}
\newcolumntype{d}[1]{D{.}{.}{#1}}

\renewcommand{\arraystretch}{1.1}

\title{Revisiting the computation of the critical points of the
  Keplerian distance}

\author{Giovanni F. Gronchi, Giulio Ba\`u, Clara Grassi}

\begin{document}
\maketitle

\begin{abstract}
We consider the Keplerian distance $d$ in the case of two elliptic
orbits, i.e. the distance between one point on the first ellipse and
one point on the second one, assuming they have a common focus. The
absolute minimum $\dmin$ of this function, called MOID or orbit
distance in the literature, is relevant to detect possible impacts
between two objects following approximately these elliptic
trajectories.  We revisit and compare two different approaches to
compute the critical points of $d^2$, where we squared the distance
$d$ to include crossing points among the critical ones. One approach
uses trigonometric polynomials, the other uses ordinary polynomials. A
new way to test the reliability of the computation of $\dmin$ is
introduced, based on optimal estimates that can be found in the
literature.  The planar case is also discussed: in this case we
present an estimate for the maximal number of critical points of
$d^2$, together with a conjecture supported by numerical tests.
\end{abstract}

\section{Introduction}

The distance $d$ between two points on two Keplerian orbits with a
common focus, that we call {\em Keplerian distance}, appears in a
natural way in Celestial Mechanics.  The absolute minimum of $d$ is
called MOID (minimum orbital intersection distance), or simply {\em
  orbit distance} in the literature, and we denote it by $\dmin$.  It
is important to be able to track $\dmin$, and actually all the local
minimum points of $d$, to detect possible impacts between two
celestial bodies following approximately these trajectories, e.g. an
asteroid with the Earth \cite{milaniAstIII, milaniLoV2005}, or two
Earth satellites \cite{rossi_2020}.  Moreover, the information given
by $\dmin$ is useful to understand observational biases in the
distribution of the known population of NEAs, see \cite{gv13}. Because
of the growing number of Earth satellites (e.g. the mega
constellations of satellites that are going to be launched
\cite{Arroyo2021}) and discovered asteroids, fast and reliable methods
to compute the minimum values of $d$ are required.

The computation of the minimum points of $d$ can be performed by
searching for all the critical points of $d^2$, where considering the
squared distance allows us to include trajectory-crossing points in
the results.

There are several papers in the literature concerning the computation
of the critical points of $d^2$, e.g. \cite{sitarski1968,
  dybczyetal1986, kholvass1999, gronchi2002, gronchi2005,
  balkhol2005}.
%
%

Some authors also propose methods for a fast computation of $\dmin$
only, e.g. \cite{wisrick13, hedoetal2018}.


We will focus on an algebraic approach
for the case of two elliptic trajectories, as in \cite{kholvass1999},
\cite{gronchi2002}.
In \cite{kholvass1999} the critical points of $d^2$ are found by
computing the roots of a trigonometric polynomial $g(u)$ of degree 8,
where $u$ is the eccentric anomaly parametrizing one of the
trajectories. The polynomial $g(u)$ is obtained by the computation of
a Groebner basis, implying that generically we can not solve this
problem by a polynomial with a smaller degree.
In \cite{gronchi2002}, resultant theory is applied to a system of two
bivariate ordinary polynomials, together with the discrete Fourier
transform, to obtain (generically) a univariate polynomial of degree
20 in a variable $t$, with a factor $(1+t^2)^2$ leading to 4 pure
imaginary roots that are discarded, so that we may have at most 16
real roots.
Note that the trigonometric polynomial $g(u)$ of degree 8 corresponds
to an ordinary polynomial of degree 16 in the variable $t$ through the
transformation $t=\tan(u/2)$.
These methods were extended to the case of unbounded conics with a
common focus in \cite{balkhol2005}, \cite{gronchi2005}.

In this paper we revisit the computation of the critical points of
$d^2$ for two elliptic trajectories by applying resultant theory to
polynomial systems written in terms of the eccentric or the true
anomalies. We obtain different methods using either ordinary or
trigonometric polynomials. Moreover, we are able to compute via
resultant theory the 8-th degree trigonometric polynomial $g(u)$ found
by \cite{kholvass1999}, and its analogue using the true anomalies (see
Sections~\ref{s:ecc_anom}, \ref{s:true_anom}).  Some numerical tests
comparing these methods are presented.
We also test the reliability of the methods by taking advantage of the
estimates for the values of $\dmin$ introduced in \cite{gv13} when one
trajectory is circular. For the case of two ellipses, since we do not
have such estimates for $\dmin$, we use the optimal bounds for the
nodal distance $\deltanod$ derived in \cite{gn20}.

After introducing some notation in Section \ref{s:prelim}, we deal
with the problem using eccentric anomalies and ordinary polynomials in
Section \ref{s:ecc_reg}. In Sections \ref{s:ecc_anom} and
\ref{s:true_anom} we describe other procedures
employing trigonometric polynomials and, respectively, eccentric or
true anomalies. Some numerical tests and the reliability of our
computations are discussed in Section \ref{s:num_tests}. Finally, we
present results for the maximum number of critical points in the
planar problem in Section \ref{s:planar}, and draw some conclusions in
Section \ref{s:conclusions}. Additional details of the computations
can be found in the Appendix.

\section{Preliminaries}
\label{s:prelim}

Let ${\cal E}_1$ and ${\cal E}_2$ be two confocal elliptic
trajectories, with ${\cal E}_i$ defined by the five Keplerian orbital
elements $a_i,e_i,i_i,\Omega_i,\omega_i$.  We introduce the Keplerian
distance
\begin{equation}
  d(V) = \sqrt{\langle {\cal X}_1-{\cal X}_2,{\cal X}_1-{\cal X}_2\rangle},
  \label{kepdistance}
\end{equation}
where ${\cal X}_1,{\cal X}_2\in\R^3$ are the Cartesian coordinates of
a point on ${\cal E}_1$ and a point on ${\cal E}_2$, corresponding to
the vector $V=(v_1,v_2)$, where $v_i$ is a parameter along the
trajectory ${\cal E}_i$.  In this paper we will parametrize the orbits
either with the eccentric anomalies $u_i$ or with the true anomalies
$f_i$.

Let $(x_1,y_1)$ and $(x_2,y_2)$ be Cartesian coordinates of two points
on the two trajectories, each in its respective plane. The origin for
both coordinate systems is the common focus of the two ellipses. We
can write
\begin{align*}
  {\mathcal X}_1 &= x_1\,{\mathcal P} + y_1\,{\mathcal Q},\\
  {\mathcal X}_2 &= x_2\,{\mathfrak p} + y_2\,{\mathfrak q},
\end{align*}
with
\begin{align*}
  {\mathcal P} &= \left(P_x\,,P_y\,,P_z\right),
  &
  {\mathcal Q} &= \left(Q_x\,,Q_y\,,Q_z\right),\\
  {\mathfrak p} &= \left(p_x\,,p_y\,,p_z\right),
  &
  {\mathfrak q} &= \left(q_x\,,q_y\,,q_z\right),
\end{align*}
where
\begin{eqnarray*}
  &&P_x = \cos\omega_1\cos\Omega_1 - \cos
  i_1\sin\omega_1\sin\Omega_1,\\
  &&P_y = \cos\omega_1\sin\Omega_1 + \cos
  i_1\sin\omega_1\cos\Omega_1,\\
  &&P_z =\sin\omega_1\sin i_1,\\
  &&Q_x = -\sin\omega_1\cos\Omega_1 - \cos
  i_1\cos\omega_1\sin\Omega_1,\\
  &&Q_y = -\sin\omega_1\sin\Omega_1 + \cos
  i_1\cos\omega_1\cos\Omega_1,\\
  &&Q_z =\cos\omega_1 \sin i_1,\\
  &&p_x = \cos\omega_2 \cos\Omega_2 - \cos
  i_2\sin\omega_2\sin\Omega_2,\\
  &&p_y = \cos\omega_2 \sin\Omega_2 + \cos
  i_2\sin\omega_2\cos\Omega_2,\\
  &&p_z = \sin\omega_2 \sin i_2,\\
  &&q_x = -\sin\omega_2 \cos\Omega_2 - \cos
  i_2\cos\omega_2\sin\Omega_2,\\
  &&q_y = -\sin\omega_2 \sin\Omega_2 + \cos
  i_2\cos\omega_2\cos\Omega_2,\\
  &&q_z = \cos\omega_2 \sin i_2 .
\end{eqnarray*}

\noindent If we use the eccentric anomalies $u_i$ we have
\[
x_i  = a_i(\cos u_i - e_i),\qquad y_i = a_i\sqrt{1-e_i^2}\sin u_i,
\]
for $i=1,2$, while with the true anomalies $f_i$ we have
\[
x_i = r_i\cos f_i,\qquad y_i = r_i\sin f_i,\qquad r_i = \frac{a_i(1-e_i^2)}{1+e_i\cos f_i}.
\]
Note that
\[
|\mathcal{P}| = |\mathcal{Q}| = |\mathfrak{p}| = |\mathfrak{q}| = 1,
\qquad
\langle\mathcal{P},\mathcal{Q}\rangle =
\langle\mathfrak{p},\mathfrak{q}\rangle = 0,
\]
and set
\[
K = \langle {\mathcal P},{\mathfrak p}\rangle,
\hskip 0.5cm
L = \langle {\mathcal Q},{\mathfrak p}\rangle,
\hskip 0.5cm
M = \langle {\mathcal P},{\mathfrak q}\rangle,
\hskip 0.5cm
N = \langle {\mathcal Q},{\mathfrak q}\rangle .
\]

\section{Eccentric anomalies and ordinary polynomials}
\label{s:ecc_reg}

We look for the critical points of the squared distance $d^2$ as a
function of the eccentric anomalies $u_1, u_2$, that is we consider
the system
\begin{equation}
  \nabla d^2(u_1,u_2) = \mathbf{0},
  \label{grad}
\end{equation}
where $\nabla d^2 = \bigl(\frac{\partial d^2}{\partial u_1},
\frac{\partial d^2}{\partial u_2}\bigr)$. We can write
\begin{equation}
  \left\{
  \begin{split}
    &\frac{1}{2}\frac{\partial d^2}{\partial u_1} = \langle\frac{\partial {\cal X}_1}{\partial u_1},{\cal X}_1-{\cal X}_2\rangle
    = \frac{\partial x_1}{\partial u_1}(x_1 - Kx_2 - My_2) + \frac{\partial y_1}{\partial u_1}(y_1 - Lx_2 - Ny_2),\\
    &\frac{1}{2}\frac{\partial d^2}{\partial u_2} = -\langle\frac{\partial {\cal X}_2}{\partial u_2},{\cal X}_1-{\cal X}_2\rangle
    = \frac{\partial x_2}{\partial u_2}(x_2 - Kx_1 - Ly_1) + \frac{\partial y_2}{\partial u_2}(y_2 - Mx_1 - Ny_1),
  \end{split}
  \right.
  \label{critpts_eccanom}
\end{equation}
where
\[
\frac{\partial x_i}{\partial u_i} = -a_i\sin u_i,\qquad
\frac{\partial y_i}{\partial u_i} = a_i\sqrt{1-e_i^2}\cos u_i,\hskip 1cm i=1,2.
\]

\noindent System \eqref{critpts_eccanom} can be written as
\begin{equation}
  \left\{
  \begin{split}
    & 2(A_1 - A_3)\sin u_1\cos u_1 + A_7\cos u_1\sin u_2 + A_8\cos u_1\cos u_2\\
    & \quad- A_9\sin u_1\sin u_2 - A_{10}\sin u_1\cos u_2 + A_{11}\cos u_1 - A_{12}\sin u_1 = 0,\\
    & 2(A_4 - A_6)\sin u_2\cos u_2 + A_7\sin u_1 \cos u_2 - A_8\sin u_1\sin u_2\\
    & \quad+ A_9\cos u_1\cos u_2 -  A_{10}\cos u_1\sin u_2 + A_{13}\cos u_2 - A_{14}\sin u_2 = 0, 
    \label{ecc_gradient}
  \end{split}
  \right.
\end{equation}
with
\[
\begin{aligned}
  & A_1 = a_1^2(1-e_1^2), &\quad &A_2 = 0,\\
  & A_3 = a_1^2, &\quad &A_4 = a_2^2(1-e_2^2),\\
  & A_5 = 0, &\quad &A_6 = a_2^2,\\ 
  & A_7 = -2a_1a_2\sqrt{1-e_1^2}\sqrt{1-e_2^2} N, &\quad &A_8 = -2a_1a_2\sqrt{1-e_1^2} L,\\
  & A_9 = -2a_1a_2\sqrt{1-e_2^2} M, &\quad &A_{10} = -2a_1a_2 K,\\
  & A_{11} = 2a_1a_2e_2\sqrt{1-e_1^2} L, &\quad &A_{12} = 2a_1(a_2e_2 K - a_1e_1),\\
  & A_{13} = 2a_1a_2e_1\sqrt{1-e_2^2} M, &\quad &A_{14} = 2a_2(a_1e_1 K - a_2e_2),\\
  & A_{15} = a_1^2 e_1^2 + a_2^2 e_2^2 - 2a_1a_2e_1e_2 K. & &
\end{aligned}
\]
Following \cite{gronchi2002}, we can transform \eqref{ecc_gradient}
into a system of two bivariate ordinary polynomials in the variables
$t,s$ through 
\begin{equation*}
  \sin u_1 = \displaystyle{2t\over 1+t^2},\qquad
  \cos u_1 = \displaystyle{1-t^2 \over 1+t^2},\qquad
  \sin u_2 = \displaystyle{2s\over 1+s^2},\qquad
  \cos u_2 = \displaystyle{1-s^2 \over 1+s^2}.
  \label{param}
\end{equation*}
Then, \eqref{ecc_gradient} becomes
\begin{equation}
  \left\{
  \begin{split}
    & p(t,s) = \alpha (t) s^2 + \beta (t) s + \gamma (t) = 0,\\
    & q(t,s) = A (t) s^4 + B (t) s^3 + D (t) s - A (t) = 0,
    \label{polysys}
  \end{split}
  \right.
\end{equation}
where
\begin{align*}
  \alpha(t) &= (A_{11}-A_8) + (4A_1 -4A_3 + 2A_{10} -2A_{12})t\\
  &\,\quad +(-4A_1 +4A_3 +2A_{10} -2A_{12})t^3 +(A_8-A_{11})t^4,\\
  \beta(t) &= 2A_7 -4A_9t -4A_9t^3 -2A_7t^4,\\
  \gamma(t) &= (A_{11}+A_8) + (4A_1 -4A_3 - 2A_{10} -2A_{12})t\\
  &\,\quad +(-4A_1 +4A_3 -2A_{10} -2A_{12})t^3-(A_8+A_{11})t^4
\end{align*}
and
\begin{align*}
  A(t) &= -(A_9 + A_{13}) - 2A_7t + (A_9 - A_{13})t^2,\\
  B(t) &= (-4A_4 + 4A_6 - 2A_{10} - 2A_{14}) - 4A_8t\\
  &\,\quad +(-4A_4 + 4A_6 + 2A_{10} - 2A_{14})t^2,\\
  D(t) &= (4A_4 - 4A_6 - 2A_{10} - 2A_{14}) - 4A_8t\\
  &\,\quad+(4A_4 - 4A_6 + 2A_{10} - 2A_{14})t^2.
\end{align*}
Let
\begin{equation}
  S_0 = \left(
  \begin{array}{ccccrr}
    \alpha  &0       &0      &0       &A   &0   \\
    \beta   &\alpha  &0      &0       &B   &A   \\
    \gamma  &\beta   &\alpha &0       &0   &B   \\
    0       &\gamma  &\beta  &\alpha  &D   &0   \\
    0       &0       &\gamma &\beta   &-A  &D  \\
    0       &0       &0      &\gamma  &0   &-A 
  \end{array}
  \right ) 
\end{equation}
be the Sylvester matrix related to \eqref{polysys}. From resultant
theory \cite{cox92} we know that the complex roots of $\det(S_0(t))$
correspond to all the $t$-components of the solutions $(t, s)\in\C^2$
of \eqref{polysys}.
The determinant $\det(S_0(t))$ is in general a polynomial of degree 20
in $t$. We notice that it can be factorized as
\[
\det\bigl(S_0(t)\bigr) = (1+t^2)^2\det\bigl(\hat{S}(t)\bigr)
\]
with
\begin{equation}
  \hat{S} = \left (
  \begin{array}{cccccc}
    \widetilde{\sigma}_1  & -\widetilde{\sigma}_2  & -\widetilde{\sigma}_1  & \phantom{-}\widetilde{\sigma}_2  & 0  & -\widetilde{\sigma}_3\\
    \widetilde{\sigma}_2  &\phantom{-}\widetilde{\sigma}_1  & -\widetilde{\sigma}_2  & -\widetilde{\sigma}_1  & \phantom{-}\widetilde{\sigma}_3  & 0\\
    \sigma_4  & \phantom{-}\sigma_2  & \phantom{-}\sigma_1  & -\sigma_2  & \phantom{-}\sigma_6  & \phantom{-}\sigma_3\\
    0  & \phantom{-}\sigma_4  & \phantom{-}\sigma_2  & \phantom{-}\sigma_1  & \phantom{-}\sigma_5  & \phantom{-}\sigma_6\\
    0  & 0  & \phantom{-}\sigma_4  & \phantom{-}\sigma_2  & -\sigma_6  & \phantom{-}\sigma_5\\
    0  & 0  & 0  & \phantom{-}\sigma_4   & 0  & -\sigma_6 
  \end{array}
  \right ),
\end{equation}
where
\begin{equation}
  \begin{aligned}
    & \sigma_1 = \alpha-\gamma,\quad
    & \sigma_2 & = \beta,\quad
    & \sigma_3 & = B-D,\cr
    & \widetilde{\sigma}_1 = \frac{\alpha-\gamma}{1+t^2},\quad
    & \widetilde{\sigma}_2 &= \frac{\beta}{1+t^2},\quad
    & \widetilde{\sigma}_3 &= \frac{B-D}{1+t^2},\cr
    & \sigma_4 = \gamma,\quad
    & \sigma_5 &= D,\quad
    & \sigma_6 &= A.\cr
  \end{aligned}
\end{equation}
Therefore, to find the $t$-components corresponding to the critical
points, we can look for the solutions of $\det(\hat S)=0$, which in
general is a polynomial equation of degree 16. We can follow the same
steps explained in \cite[Sect. 4.3]{gronchi2005} to obtain the
coefficients of the polynomial $\det(\hat S)$ by an
evaluation/interpolation procedure based on the discrete Fourier
transform. Then, the method described in \cite{bini} is applied to
compute its roots. We substitute each of the real roots $t$ of
$\det(\hat S)$ in \eqref{polysys} and use the first equation
$p(t,s)=0$ to compute the two possible values of the $s$
variable. Finally, we evaluate $q$ at these points $(t,s)$ and choose
the value of $s$ that gives the evaluation with the smallest absolute
value.

\subsection{Angular shifts}
\label{s:angshift_ecc}
Define a shifted angle $v_1=u_1-s_1$ and let
\begin{equation}
  \sin v_1 = \frac{2z}{1+z^2},\qquad \cos v_1 = \frac{1-z^2}{1+z^2}.
  \label{zparam}
\end{equation}
Then, system \eqref{ecc_gradient} becomes
\begin{equation}
  \left\{
  \begin{split}
    &\tilde p(z,s) = \tilde\alpha (z) s^2 + \tilde\beta (z) s + \tilde\gamma (z) = 0,\\
    &\tilde q(z,s) = \tilde A (z) s^4 + \tilde B (z) s^3 + \tilde D (z) s - \tilde A (z) = 0.
    \label{polysys_shift}
  \end{split}
  \right.
\end{equation}
The coefficients $\tilde\alpha, \tilde\beta, \tilde\gamma, \tilde A,
\tilde B, \tilde D$ are written in Appendix \ref{a:shift_ord}. If
$T_0(z)$ is the Sylvester matrix related to \eqref{polysys_shift} we
get
\[
\det(T_0(z)) = (1+z^2)^2\det(\hat T(z)),
\]
with
\begin{equation*}
  \hat{T} = \left(
  \begin{array}{cccccc}
    \widetilde{\tau}_1  & -\widetilde{\tau}_2  & -\widetilde{\tau}_1  & \phantom{-}\widetilde{\tau}_2  & 0  & -\widetilde{\tau}_3\\
    \widetilde{\tau}_2  & \phantom{-}\widetilde{\tau}_1  & -\widetilde{\tau}_2  & -\widetilde{\tau}_1  & \phantom{-}\widetilde{\tau}_3  & 0\\
    \tau_4  & \phantom{-}\tau_2  & \phantom{-}\tau_1  & -\tau_2  & \phantom{-}\tau_6  & \phantom{-}\tau_3\\
    0  & \phantom{-}\tau_4  & \phantom{-}\tau_2  & \phantom{-}\tau_1  & \phantom{-}\tau_5  & \phantom{-}\tau_6\\
    0  & 0  & \phantom{-}\tau_4  & \phantom{-}\tau_2  & -\tau_6  & \phantom{-}\tau_5 \\
    0  & 0  & 0  & \phantom{-}\tau_4  & 0  & -\tau_6 
  \end{array}
  \right),
\end{equation*}
where
\begin{equation}
  \begin{aligned}
    & \tau_1 = \tilde\alpha-\tilde\gamma,\quad
    & \tau_2 &= \tilde\beta,\quad
    & \tau_3 &= \tilde B-\tilde D,\cr
    & \widetilde{\tau}_1 = \frac{\tilde\alpha-\tilde\gamma}{1+z^2},\quad
    & \widetilde{\tau}_2 &= \frac{\tilde\beta}{1+z^2},\quad
    & \widetilde{\tau}_3 &= \frac{\tilde B-\tilde D}{1+z^2},\cr
    & \tau_4 = \tilde\gamma,\quad
    & \tau_5 &= \tilde D,\quad
    & \tau_6 &= \tilde A.\cr
    \end{aligned}
\end{equation}
We find the values of $z$ by solving the polynomial equation
$\det(\hat T)=0$, which again has generically degree 16. We compute
the values of $v_1$ from \eqref{zparam} and shift back to obtain the
$u_1$ components of the critical points. Substituting in
\eqref{ecc_gradient} and applying the angular shift $u_2=v_2+s_2$, we
consider the system
\[
\left\{
\begin{split}
  & \mathsf{A}\cos v_2 + \mathsf{B}\sin v_2 + \mathsf{C} = 0,\\
  & \cos^2v_2 + \sin^2v_2 -1 = 0,
\end{split}
\right.
\]
where the first equation corresponds to the first equation in
\eqref{ecc_gradient}, and
\[
\begin{aligned}
  & \mathsf{A} = (A_8\cos u_1-A_{10}\sin u_1)\cos s_2 + (A_7\cos u_1-A_9\sin u_1)\sin s_2,\\
  & \mathsf{B} = -(A_8\cos u_1-A_{10}\sin u_1)\sin s_2 + (A_7\cos u_1-A_9\sin u_1)\cos s_2,\\
  & \mathsf{C} = 2(A_1-A_3)\sin u_1\cos u_1+A_{11}\cos u_1-A_{12}\sin u_1.
\end{aligned}
\]
For each value of $u_1$ we compute two solutions for $\cos v_2$, $\sin
v_2$ and the corresponding values of $\cos u_2$, $\sin u_2$. We choose
between them by substituting in the second equation in
\eqref{ecc_gradient}.

\section{Eccentric anomalies and trigonometric polynomials}
\label{s:ecc_anom}

To work with trigonometric polynomials, we write system
\eqref{critpts_eccanom} as
\begin{equation}
  \left\{
  \begin{split}
    & \lambda\sin u_1\cos u_1 + \mu\cos u_1 + \nu\sin u_1 = 0,\cr
    &\alpha\cos u_1 + \beta\sin u_1 + \gamma = 0,\cr
  \end{split}
  \right.
  \label{u_critpts}
\end{equation}
where
\[
\begin{split}
  \lambda &=   a_1e_1^2,\cr
  \mu &= a_2\sqrt{1-e_1^2}\Bigl(\sqrt{1-e_2^2}N\sin u_2 + L\cos u_2 - e_2  L\Bigr),\cr
  \nu &= a_2e_2K - a_1e_1 - a_2\sqrt{1-e_2^2}M\sin u_2 - a_2K\cos u_2,\cr
  \alpha &= a_1\Bigl(\sqrt{1-e_2^2}M\cos u_2 - K\sin u_2\Bigr),\cr
  \beta &= a_1\sqrt{1-e_1^2}\Bigl(\sqrt{1-e_2^2}N\cos u_2 - L\sin u_2\Bigr),\cr
  \gamma &= a_2e_2^2\sin u_2\cos u_2 - a_1e_1\sqrt{1-e_2^2}M\cos u_2 + (a_1e_1K - a_2e_2)\sin u_2.\cr
\end{split}
\]
Inserting relation
\begin{equation}
  \sin u_1 = -\frac{1}{\beta}(\alpha\cos u_1 + \gamma)
  \label{eqforsinu1}
\end{equation}
into $\cos^2u_1+\sin^2u_1-1 = 0$ and into the first equation in
\eqref{u_critpts}, we obtain
\begin{equation}
  \left\{
  \begin{split}
    & (\alpha^2+\beta^2)\cos^2 u_1 + 2\alpha\gamma\cos u_1 + \gamma^2-\beta^2 = 0,\cr
    & - \alpha\lambda\cos^2 u_1 + (\beta\mu - \lambda\gamma-\alpha\nu)\cos u_1 -\gamma\nu = 0.\cr
  \end{split}
  \right.
  \label{cosUeqs}
\end{equation}
We call $p_1$, $p_2$ the two trigonometric polynomials appearing on
the left-hand side of \eqref{cosUeqs}.
The Sylvester matrix of $p_1$ and $p_2$ is
\[
\mathscr{S} = \left[
  \begin{array}{cccc}
    \alpha^2+\beta^2  & 0  & -\alpha\lambda  & 0\cr
    2\alpha\gamma  & \alpha^2+\beta^2  & \beta\mu - \lambda\gamma-\alpha\nu  & -\alpha\lambda\cr
    \gamma^2-\beta^2  & 2\alpha\gamma  & -\gamma\nu  & \beta\mu-\lambda\gamma-\alpha\nu\cr
    0  &\gamma^2-\beta^2  & 0  &-\gamma\nu\cr
  \end{array}
  \right].
\]
We define
\[
\mathscr{G}(u_2) = \det\mathscr{S}(u_2),
\]
which corresponds to the resultant of $p_1$, $p_2$ with respect to
$\cos u_1$ and is a trigonometric polynomial in $u_2$ only. The $u_2$
component of each critical point satisfies $\mathscr{G}(u_2)=0$.

\begin{proposition}
  \label{prop:ea_factor}
  We can extract a factor $\beta^2$ from $\det\mathscr{S}$.
\end{proposition}

\begin{proof}
Using simple properties of determinants, we can write
$\det\mathscr{S}$ as a sum of different terms. The terms independent
from $\beta$ in this sum are given by 
\[
\small
\begin{split}
  &\left|
  \begin{array}{cccc}
    \alpha^2  & 0  & -\alpha\lambda  & 0\cr
    2\alpha\gamma  & \alpha^2  & -\lambda\gamma-\alpha\nu  & -\alpha\lambda\cr
    \gamma^2  & 2\alpha\gamma  & -\gamma\nu  & -\lambda\gamma-\alpha\nu\cr
    0  & \gamma^2  & 0  & -\gamma\nu\cr
  \end{array}
  \right|
  \cr
  &={\left|
    \begin{array}{cccc}
      \alpha^2  & 0  &\alpha\lambda  & 0\cr
      \alpha\gamma  & 0  & \lambda\gamma  & 0\cr
      \gamma^2  & 2\alpha\gamma  & \gamma\nu  & \lambda\gamma+\alpha\nu\cr
      0  & \gamma^2  & 0  & \gamma\nu\cr
    \end{array}
    \right|}
  +\left|
  \begin{array}{cccc}
    \alpha^2  & 0  & \alpha\lambda  & 0\cr
    \alpha\gamma  & \alpha^2  & \alpha\nu  & \alpha\lambda\cr
    \gamma^2  & 2\alpha\gamma  & \gamma\nu  & \lambda\gamma+\alpha\nu\cr
    0  & \gamma^2  & 0  & \gamma\nu\cr
  \end{array}
  \right|\cr
  &=
  \left|
  \begin{array}{cccc}
    \alpha^2  & 0  & \alpha\lambda  & 0\cr
    \alpha\gamma  & \alpha^2  & \alpha\nu  & \alpha\lambda\cr
    0  & \alpha\gamma  & 0  & \alpha\nu\cr
    0  & \gamma^2  & 0  & \gamma\nu\cr
  \end{array}
  \right|
  + \left|
  \begin{array}{cccc}
    \alpha^2  & 0  & \alpha\lambda  & 0\cr
    \alpha\gamma  & \alpha^2  & \alpha\nu  & \alpha\lambda\cr
    \gamma^2  & \gamma\alpha  & \gamma\nu  & \gamma\lambda\cr
    0  & \gamma^2  & 0  & \gamma\nu\cr
  \end{array}
  \right|
  ,\cr
\end{split}
\]
and both determinants are $0$. The linear terms in $\beta$ are given
by
\[
\small
\begin{split}
  &\left|
  \begin{array}{cccc}
    \alpha^2  & 0  & -\alpha\lambda  & 0\cr
    2\alpha\gamma  & \alpha^2  & -\lambda\gamma-\alpha\nu  & 0\cr
    \gamma^2  & 2\alpha\gamma  & -\gamma\nu  & \beta\mu\cr
    0  & \gamma^2  & 0  & 0\cr
  \end{array}
  \right|
  +
  \left|
  \begin{array}{cccc}
    \alpha^2  & 0  & 0  & 0\cr
    2\alpha\gamma  & \alpha^2  & \beta\mu  & -\alpha\lambda\cr
    \gamma^2  & 2\alpha\gamma  & 0  & -\lambda\gamma-\alpha\nu\cr
    0  & \gamma^2  & 0  & -\gamma\nu\cr
  \end{array}
  \right|,
\end{split}
\]
and this sum is $0$, because the two determinants are opposite.
Therefore, $\mathscr{G}(u_2)$ is made by terms
of order higher than $1$ in $\beta$. It results
\[
\mathscr{G}(u_2) = \mathfrak{D}_1 + \mathfrak{D}_2 + \mathfrak{D}_3,
\]
where
\[
\small
\begin{split}
  &\mathfrak{D}_1 = \left|
  \begin{array}{cccc}
    \beta^2  & 0  & -\alpha\lambda  & 0\cr
    0  & \beta^2  & \beta\mu-\lambda\gamma-\alpha\nu  & -\alpha\lambda\cr
    -\beta^2  & 0  & -\gamma\nu  & \beta\mu-\lambda\gamma-\alpha\nu\cr
    0  & -\beta^2  & 0  & -\gamma\nu\cr
  \end{array}
  \right|,\cr
  &\mathfrak{D}_2 = 
  \left|
  \begin{array}{cccc}
    \beta^2  & 0  & -\alpha\lambda  & 0\cr
    0  & \alpha^2  & \beta\mu-\lambda\gamma-\alpha\nu  & -\alpha\lambda\cr
    -\beta^2  & 2\alpha\gamma  & -\gamma\nu & \beta\mu-\lambda\gamma-\alpha\nu\cr
    0  & \gamma^2  & 0  & -\gamma\nu\cr
  \end{array}
  \right|
  +\left|
  \begin{array}{cccc}
    \alpha^2  & 0  & -\alpha\lambda  & 0\cr
    2\alpha\gamma  & \beta^2  & \beta\mu-\lambda\gamma-\alpha\nu  & -\alpha\lambda\cr
    \gamma^2  & 0  & -\gamma\nu  & \beta\mu-\lambda\gamma-\alpha\nu\cr
    0  & -\beta^2  & 0  & -\gamma\nu\cr
  \end{array}
  \right|,\cr
  &\mathfrak{D}_3 = \left|
  \begin{array}{cccc}
    \alpha^2  & 0  & 0  & 0\cr
    2\alpha\gamma  & \alpha^2  & \beta\mu  & 0\cr
    \gamma^2  & 2\alpha\gamma  & 0  & \beta\mu\cr
    0  & \gamma^2  & 0  & 0\cr
  \end{array}
  \right|.\cr
\end{split}
\]
Their explicit expressions read
\[
\small
\begin{split}
  &\mathfrak{D}_1 = \beta^4\bigl[(\alpha\lambda+\gamma\nu)^2-(\beta\mu-\lambda\gamma-\alpha\nu)^2\bigr],\cr
  &\mathfrak{D}_2 = \beta^2\left[2\alpha\beta\gamma\mu(\gamma\nu-\alpha\lambda)
    + ((\beta\mu-\lambda\gamma-\alpha\nu)^2-4\alpha\gamma\lambda\nu)(\gamma^2-\alpha^2)\right],\cr
  &\mathfrak{D}_3 = \alpha^2\beta^2\gamma^2\mu^2.
\end{split}
\]


\end{proof}

\noindent The trigonometric polynomial
\begin{align*}
  g(u_2) &= \mathscr{G}(u_2)/\beta^2\\
  &=\gamma^4\lambda^2-\beta^4\mu^2-\alpha^4\nu^2+2\,\mu\,\lambda\,\gamma\,\beta\,\left(\beta^2-\gamma^2\right)
  +2\,\mu\,\nu\,\alpha\,\beta\,\left(\alpha^2+\beta^2\right)\\
  &\,\quad +2\,\lambda\,\nu\,\gamma\,\alpha\,\left(\alpha^2-\gamma^2\right)+\left(-\lambda^2+\mu^2+\nu^2\right)
  \left(-\alpha^2\beta^2+\gamma^2\alpha^2+\gamma^2\beta^2\right)
\end{align*}
has total degree 8 in the variables $\cos u_2$, $\sin u_2$, and
corresponds to the polynomial $g$ introduced in \cite{kholvass1999}
with Groebner bases theory. For this reason, generically, there is no
polynomial of smaller degree giving all the $u_2$ components of the
critical points of $d^2$.


We now explain the procedure to reduce the problem to the computation
of the roots of a univariate polynomial. We set
\[
\mathfrak{g}(x,y) = g(u_2),
\]
where
\[
x=\cos u_2, \qquad y=\sin u_2.
\]
We find that
\[
\mathfrak{g}(x,y) = \sum_{j=0}^6 g_j(x) y^j
\]
for some polynomial coefficients $g_j$ such that
\[
\begin{split}
&\deg g_0 =\deg g_1 = \deg g_2 = 6,\\
&\deg g_3 = 5, \quad \deg g_4 = 4, \quad \deg g_5 = 3, \quad \deg g_6 = 2.
\end{split}
\]
Then, we consider the polynomial system
\begin{equation}
  \left\{
  \begin{split}
    &\mathfrak{g}(x,y) = 0,\cr
    & x^2+y^2-1=0 .\cr
  \end{split}
  \right.
  \label{gsysecc}
\end{equation}
Using relations
\[
y^{2k} = (1-x^2)^k,\qquad y^{2k+1} = y(1-x^2)^k,\qquad k\in\mathbb{N},
\]
obtained from the second equation in \eqref{gsysecc}, we can
substitute $\mathfrak{g}$ in system \eqref{gsysecc} with
\[
\tilde{\mathfrak{g}}(x,y) =  a(x)y + b(x), 
\]
where
\[
\begin{split}
  a(x) &= g_1(x) + (1-x^2)g_3(x) + (1-x^2)^2g_5(x),\cr
  b(x) &= g_0(x) + (1-x^2)g_2(x) + (1-x^2)^2g_4(x) + (1-x^2)^3g_6(x).\cr
\end{split}
\]
We can also write
\[
\begin{split}
  a(x) &= a_0(x) + x^2 a_2(x) + x^4 a_4(x),\cr
  b(x) &= b_0(x) + x^2b_2(x) + x^4b_4(x) + x^6b_6(x),\cr
\end{split}
\]
with
\[
\begin{aligned}
  & a_0 = g_1+g_3+g_5,  & a_2 &= -g_3-2g_5,  & a_4 &= g_5,  &  &\cr
  & b_0 = g_0+g_2+g_4+g_6,  & b_2 &= -g_2-2g_4-3g_6,  & b_4 &= g_4+3g_6,  & b_6 &= -g_6.\cr
\end{aligned}
\]
Note that $a$ and $b$ have degree 7 and 8, respectively. We eliminate
$y$ from system
\[
\left\{
\begin{split}
  & \tilde{\mathfrak{g}}(x,y) = 0,\cr
  & x^2+y^2 -1 = 0\cr
\end{split}
\right.
\]
by computing the resultant $\uvp(x)$ of the two polynomials with
respect to $y$, and obtain
\begin{equation}
  \uvp(x) =
  \left|
  \begin{array}{ccc}
    a(x)  & 0  & 1\cr
    b(x)  & a(x)  & 0\cr
    0  & b(x)  & x^2-1\cr 
  \end{array}
  \right| = a^2(x)(x^2-1) + b^2(x),
  \label{deg16u}
\end{equation}
which is a univariate polynomial of degree 16.


Each of the real roots $x$ of $\uvp$, with $|x|\le 1$, is substituted
into the equation $\tilde{\mathfrak{g}}(x,y) = 0$ to get the value of
$y$. Finally, we evaluate $\alpha,\beta,\gamma,\mu,\nu$ at the
computed pairs $(x,y)$ and solve system \eqref{u_critpts} by computing
the values of $\cos u_1$ and $\sin u_1$ from \eqref{cosUeqs} and
\eqref{eqforsinu1}, respectively.

\subsection{Finding the roots of $\uvp$ with Chebychev's polynomials}

To compute the roots of the polynomial $\uvp(x)$ in a numerically
stable way, we need to express $\uvp$ in a basis ensuring that the
roots are well-conditioned functions of its coefficients. This can be
achieved using Chebyshev's polynomials \cite{nofper2017} in place of
the standard monomial basis.

In the monomial basis we have
\begin{equation}
  \uvp(x) = \sum_{j=0}^np_j x^j,
  \label{p_poly}
\end{equation}
for some coefficients $p_j$. The same polynomial can be written as
\begin{equation}
  \uvp(x) = \sum_{j=0}^nc_j T_j(x),
\label{c_poly}
\end{equation}
where $T_j$ are Chebyshev's polynomials, recursively defined by
\begin{equation}
  T_0(x) = 1,\qquad T_1(x) = x,\qquad
  T_{j+1} = 2xT_j - T_{j-1}, \ j=1,\ldots,n-1,
\label{chebdef}
\end{equation}
which are a basis for the vector space of polynomials of degree at
most $n$. The coefficients $c_j$ are obtained from the $p_j$ as
follows. Setting
\[
X = (1, x, x^2,\ldots, x^n)^t,\qquad
Y = (T_0(x), T_1(x),\ldots, T_n(x))^t
\]
we have
\begin{equation}
  {A}X = Y,
  \label{AXeqY}
\end{equation}
with
\[
A = \left[
  \begin{array}{cccc}
    a_{00}  & 0  & \hdots  & 0\cr
    a_{10}  & a_{11}  & \ddots  & \vdots\cr
    \vdots  &  & \ddots  & 0\cr
    a_{n0}  & a_{n1}  & \hdots  & a_{nn}\cr
  \end{array}
  \right],
\]
where the integer coefficients $a_{ij}$ are determined from relations
\eqref{chebdef}. We invert $A$ by the following procedure. Define
\[
\tilde{A} =
\left[
  \begin{array}{cccc}
    1  & 0  & \hdots  & 0\cr
    \frac{a_{10}}{a_{11}}  & 1  & \ddots  & \vdots\cr
    \vdots  &  & \ddots  & 0\cr
    \frac{a_{n0}}{a_{nn}}  & \frac{a_{n1}}{a_{nn}}  & \hdots  & 1\cr
  \end{array}
  \right],
\qquad
\tilde{Y} = \left(\frac{T_0}{a_{00}}, \frac{T_1}{a_{11}},\ldots,
\frac{T_n}{a_{nn}}\right)^t.
\]
Equation \eqref{AXeqY} becomes
\[
\tilde{A}X = \tilde{Y}, 
\]
with
\[
\tilde{A} = I + N,
\]
where $N^n=0$, that is $N$ is a nilpotent matrix of order $n$.
Relation
\[
(I+N)(I-N+N^2-N^3+\ldots+(-1)^{n-1}N^{n-1}) = I,
\]
implies that the inverse of $\tilde{A}$ is
\begin{equation}
  \tilde{A}^{-1} = I + \sum_{j=1}^{n-1}(-1)^jN^j.
  \label{Atildeinv}
\end{equation}
Let us introduce the vectors
\[
P = (p_0, p_1, \ldots, p_n)^t, \qquad C = (c_0,c_1,\ldots,c_n)^t
\]
made by the coefficients of the polynomials in \eqref{p_poly},
\eqref{c_poly}, and the diagonal matrix
\[
D = \mathrm{diag}\{
a_{00}^{-1},
a_{11}^{-1},\ldots,
a_{nn}^{-1}\}.
\]
From \eqref{AXeqY} and \eqref{Atildeinv} we can write
\[
X = \tilde{A}^{-1}DY,
\]
so that
\[
\uvp(x) = C^tY = P^tX = (P^t\tilde{A}^{-1}D)Y.
\]
Therefore, the relation between the coefficients $c_j$ and $p_j$ is
given by
\[
C = (D\tilde{A}^{-t}) P.
\]

Searching for the roots of $\uvp(x)$ corresponds to computing the
eigenvalues of an $n\times n$ matrix $\mathscr{C}$, called {\em
  colleague matrix} \cite{good1961}. We use the form of the colleague
matrix described in \cite{casrob2021}:
\begin{equation}
  \coll =\frac{1}{2}\left[
    \begin{array}{ccccc}
      0       & 1       & 0       & \hdots    & 0       \cr
      1       & 0       & \ddots  &           & \vdots  \cr
      0       & \ddots  & \ddots  & 1         & 0       \cr      
      \vdots  & \ddots  & 1       & 0         & \sqrt{2}\cr
      0       & \hdots  & 0       & \sqrt{2}  & 0       \cr
    \end{array}
    \right]
  - \frac{1}{2c_n}
  \left[
    \begin{array}{c}
      1\cr
      0\cr
      \vdots\cr
      0\cr
    \end{array}
    \right]
  \left[
    \begin{array}{ccccc}
      c_{n-1}  & c_{n-2}  & \hdots  & c_1  & \sqrt{2}c_0\cr
    \end{array}
    \right].
  \label{colleague}
\end{equation}
The computation of the roots of a polynomial using the colleague
matrix and a backward stable eigenvalue algorithm, such as the QR
algorithm, is backward stable, provided that the 2-norm of the
polynomial is moderate (see \cite{nofper2017}).

\section{True anomalies and trigonometric polynomials}
\label{s:true_anom}

The same steps described in Section \ref{s:ecc_anom} can be applied to
look for the critical points of the squared distance function
expressed in terms of the true anomalies $f_1$, $f_2$. Note that using
true anomalies allows to deal with both bounded and unbounded
trajectories \cite{sitarski1968, gronchi2005}.


%
%
%
%
%
%
%

We write the system
\[
\nabla d^2(f_1,f_2) =
\mathbf{0}
\]
as
\begin{equation}
  \left\{
  \begin{split}
    & \alpha\cos f_1 + \beta\sin f_1 + \gamma = 0,\cr
    & \kappa\cos^2 f_1 + \lambda\sin f_1\cos f_1 + \mu\cos f_1 + \nu\sin f_1 + \kappa = 0,
  \end{split}
  \right.
  \label{critpts}
\end{equation}
where
\[
\begin{split}
  \alpha &= p_1(1+e_2\cos f_2)(K\sin f_2 - M(e_2+\cos f_2)) + p_2e_1e_2\sin f_2, \cr
  \beta &=  p_1(1+e_2\cos f_2)(L\sin f_2 - N(e_2+\cos f_2)), \cr
  \gamma &= p_2e_2\sin f_2,\cr
  \kappa &= -p_2e_1(L\cos f_2 + N\sin f_2), \cr
  \lambda &=  p_2e_1(K\cos f_2 + M\sin f_2), \cr
  \mu &= -p_2(1+e_1^2)(L\cos f_2 + N\sin f_2),\cr
  \nu &= p_1e_1(1+e_2\cos f_2) + p_2(K\cos f_2 + M\sin f_2).\cr
\end{split}
\]
We also set
\[
\begin{split}
  &\tilde{\kappa} = p_2(L\cos f_2+N\sin f_2),\cr
  &\tilde{\lambda} = p_2(K\cos f_2+M\sin f_2),
\end{split}
\]
so that
\[
\kappa = -e_1\tilde{\kappa} = \frac{e_1}{1+e_1^2}\mu,\qquad
\lambda = e_1\tilde{\lambda}, \qquad\nu = p_1e_1(1+e_2\cos f_2) + \tilde{\lambda}. 
\]
Inserting relation
\begin{equation}
  \sin f_1 = -\frac{1}{\beta}(\alpha\cos f_1 + \gamma)
  \label{sinf1}
\end{equation}
into $\cos^2f_1+\sin^2f_1-1 = 0$ and into the second equation of
\eqref{critpts}, we obtain
\begin{equation}
  \left\{
  \begin{split}
    & (\alpha^2+\beta^2)\cos^2 f_1 + 2\alpha\gamma\cos f_1 + \gamma^2-\beta^2 = 0,\cr
    & (\beta\kappa - \alpha\lambda)\cos^2 f_1 + (\beta\mu - \lambda\gamma-\alpha\nu)\cos f_1 + \beta\kappa-\gamma\nu = 0.\cr
  \end{split}
  \right.
  \label{cosVeqs}
\end{equation}
As in Section~\ref{s:ecc_anom}, we consider the Sylvester matrix of
the two polynomials in \eqref{cosVeqs}
\[
\mathscr{T} = \left[
  \begin{array}{cccc}
    \alpha^2+\beta^2  & 0  & \beta\kappa-\alpha\lambda  & 0\cr
    2\alpha\gamma  & \alpha^2+\beta^2  & \beta\mu-\lambda\gamma-\alpha\nu  &\beta\kappa-\alpha\lambda\cr
    \gamma^2-\beta^2  & 2\alpha\gamma  & \beta\kappa-\gamma\nu  & \beta\mu-\lambda\gamma-\alpha\nu\cr
    0  & \gamma^2-\beta^2  & 0  & \beta\kappa-\gamma\nu\cr
  \end{array}
  \right]
\]
and define 
\[
\mathscr{H}(f_2) = \det\mathscr{T},
\]
that we are able to factorize. In particular, we can write
\[
\mathscr{H}(f_2) = (1+e_2\cos f_2)^2\beta^2h(f_2), 
\]
where
\begin{equation}
  \begin{aligned}
    h(f_2) &= \tilde{\beta}^4\xi^2\mu^2(4\eta^2-1)+2\tilde{\beta}^3\xi\mu\left[\lambda\gamma+\alpha\nu-2\eta(\alpha\lambda+\gamma\nu)\right]\\
    &\,\quad +\tilde{\beta}^2(\alpha^2-\gamma^2)\left[\lambda^2-\nu^2+\mu^2(4\eta^2-1)\right]\\
    &\,\quad -2\mu\tilde{\beta}\tilde{\alpha}^2\xi^2\left[\tilde{\alpha}(\lambda\eta-\nu)-3\eta e_1^3 p_1\gamma\right]
    +\mu^2\tilde{\alpha}^2\left[\gamma(1-2\eta e_1)-\eta\tilde{\alpha}\xi\right]^2\\
    &\,\quad -2\mu\tilde{\beta}p_1e_1(1-e_1^2)\eta\gamma^2\left[3e_1\tilde{\alpha}\xi-\gamma(1-e_1^2)\right]
    -(\alpha^2-\gamma^2)(\nu\tilde{\alpha}+p_1e_1^2\gamma)^2,
  \end{aligned}
  \label{truean_hpoly}
\end{equation}
with
\begin{align*}
  \eta &= \frac{e_1}{1+e_1^2},  & \xi &= 1+e_2\cos f_2,\\
  \tilde{\alpha} &= p_1[K\sin f_2-M(e_2+\cos f_2)],  & \tilde{\beta} &= p_1[L\sin f_2-N(e_2+\cos f_2)].
\end{align*}
We can show that $h(f_2)$ has degree 8 in $(\cos f_2, \sin f_2)$. The
related computations are displayed in Appendix~\ref{a:true_anom}.

Let us set
\[
\mathfrak{h}(x,y) = h(f_2),
\]
where
\[
x=\cos f_2, \qquad y = \sin f_2.
\]
We find that
\[
\mathfrak{h}(x,y) = \sum_{j=0}^6 h_j(x) y^j
\]
for some polynomial coefficients $h_j$ such that
\[
\begin{split}
& \deg h_0 = 8,\quad \deg h_1 = 7,\quad \deg h_2 = 6,\\
& \deg h_3 = 5,\quad \deg h_4 = 4,\quad \deg h_5 = 3,\quad \deg h_6 = 2.
\end{split}
\]
Then, we consider the system
\begin{equation}
  \left\{
  \begin{split}
    & \mathfrak{h}(x,y) = 0,\cr
    & x^2+y^2-1 = 0.\cr
  \end{split}
  \right.
  \label{gsys}
\end{equation}
Proceeding as in Section~\ref{s:ecc_anom} we can substitute
$\mathfrak{h}(x,y)$ with
\begin{equation}
  \tilde{\mathfrak{h}}(x,y) = \mathfrak{a}(x)y + \mathfrak{b}(x), 
  \label{htilde}
\end{equation}
with
\[
\begin{split}
  \mathfrak{a}(x) &= a_0(x)+x^2a_2(x)+x^4a_4(x),\cr
  \mathfrak{b}(x) &= b_0(x)+x^2b_2(x)+x^4b_4(x)+x^6b_6(x),\cr
\end{split}
\]
where
\[
\begin{aligned}
  a_0 &= h_1+h_3+h_5,  & a_2 &= -h_3-2h_5,  & a_4 &= h_5,  &  &\\
  b_0 &= h_0+h_2+h_4+h_6,  & b_2 &= -h_2-2h_4-3h_6,  & b_4 &= h_4+3h_6,  & b_6 &=-h_6.
\end{aligned}
\]
We apply resultant theory to eliminate the dependence on $y$ as in
Section \ref{s:ecc_anom} and obtain a univariate polynomial
$\mathfrak{v}$ of degree 16. The real roots of $\mathfrak{v}$ with
absolute value $\leq 1$
correspond to the values of $\cos f_2$ we are searching for. We
compute $\sin f_2$ from \eqref{htilde} and substitute $\cos f_2$ and
$\sin f_2$ in \eqref{cosVeqs}. Finally, $\cos f_1$ and $\sin f_1$ are
found by solving \eqref{cosVeqs} and using \eqref{sinf1}.

\subsection{Angular shifts}
\label{s:angshift}

Also for the method presented in this section we consider the
application of an angular shift. If we define the new shifted angle
$v_2$ by
\[
v_2 = f_2-s_2,
\]
for some $s_2\in[0,2\pi)$, the coefficients of the polynomial
  \eqref{truean_hpoly} written in terms of $v_2$ are derived following
  the computations of Appendix \ref{a:shift}. Then, following a
  procedure analogous to Section \ref{s:true_anom}, we find the values
  of $v_2$ and shift back to get $f_2$. Finally, we can apply an
  angular shift also to the angle $f_1$ when solving system
  \eqref{cosVeqs}. Defining the shifted angle as
\[
v_1 = f_1-s_1,
\]
for $s_1\in[0,2\pi)$, system \eqref{cosVeqs} becomes
  \begin{equation}
    \left\{
    \begin{split}
      & A\cos v_1 + B\sin v_1 + C = 0,\cr
      & D\cos^2 v_1 + E\sin v_1\cos v_1 + F\cos v_1 + G\sin v_1 + H = 0,
    \end{split}
    \right.
  \end{equation}
where
\[
\begin{split}
  A &= \alpha\cos s_1+\beta\sin s_1,\cr
  B &= \beta\cos s_1-\alpha\sin s_1,\cr
  C &= \gamma,\cr
  D &= \kappa\cos^2s_1-\kappa\sin^2s_1+2\lambda\sin s_1\cos s_1,\cr
  E &= \lambda\cos^2s_1-\lambda\sin^2s_1-2\kappa\sin s_1\cos s_1,\cr
  F &= \mu\cos s_1+\nu\sin s_1,\cr
  G &= \nu\cos s_1-\mu\sin s_1,\cr
  H &= \kappa\sin^2s_1-\lambda\sin s_1\cos s_1+\kappa,
\end{split}
\]
with $\alpha, \beta, \gamma, \kappa, \lambda, \mu, \nu$ defined at the
beginning of this section.

\section{Numerical tests}
\label{s:num_tests}

We have developed Fortran codes for each of the methods presented in
this paper. We denote these methods with (OE, OES, TE, TEC, TT, TTS),
see Table~\ref{tab:err_data}. Moreover, we denote by (OT) the method
presented in \cite{gronchi2005}. Numerical tests have been carried out
for pairs of bounded trajectories to compare the different methods.

Taking the NEA catalogue available at
\url{https://newton.spacedys.com/neodys/}, we applied these methods to
compute the critical points of the squared distance between each NEA
and the Earth, and between all possible pairs of NEAs. We applied a
few simple checks to detect errors in the results:
\begin{itemize}
\item Weierstrass check (W): for each pair of trajectories we have to
  find at least one maximum and one minimum point;
  
\item Morse check (M): for each pair of trajectories, let $N$ be the
  total number of critical points, and $M$ and $m$ be the number of
  maximum and minimum points, respectively. Then (assuming $d^2$ is a
  Morse function) we must have $N=2(M+m)$;
  
\item Minimum distance check ($\dmin$): we sample the two trajectories
  with $k$ uniformly distributed points each (we used $k=10$), and
  compute the distance between each pair of points. We check that the
  minimum value of $d$ computed through this sampling is greater than
  the value of $\dmin$ obtained with our methods.
\end{itemize}
For each method a small percentage of cases fail due to some of the
errors above. However, in our tests, for each pair of orbits, at least
one method passes all three checks.

The angular shifts (see Sections \ref{s:angshift_ecc},
~\ref{s:angshift}) allow us to solve the majority of detected errors
for the methods of Sections \ref{s:ecc_reg} and \ref{s:true_anom}
without shift (OE, TT). Applying a shift could also be a way to solve
most of the errors detected by the method of Section \ref{s:ecc_anom}
(TE, TEC).

Some data on the detected errors for each method is reported in Table
\ref{tab:err_data}. Here we show the percentages of cases failing some
of the three checks described above. We note that the NEA catalogue
contains 31,563 NEAs with bounded orbits (to the date of February 25,
2023). Therefore, for the NEA--Earth test we are considering 31,563
pairs of orbits, while for the NEA--NEA test the number of total pairs
is 498,095,703. From Table \ref{tab:err_data} we see that the method
TE
is improved with Chebychev's polynomials (TEC). In the same way, the
methods OE, TT
are improved by applying angular shifts in case of detected errors
(OES, TTS). Indeed, the method TTS
turns out to be the most reliable for the computation of $\dmin$.

\begin{table}
  \small
  \centering
  \renewcommand{\arraystretch}{1.2}
  \begin{tabular}{p{4cm} r|ccc|ccc}
   \hskip 1.5cm\textbf{algorithm}  & \multicolumn{1}{r|}{}  & \multicolumn{3}{c|}{\textbf{NEA -- Earth}}  & \multicolumn{3}{c}{\textbf{NEA -- NEA}}\\
   & \multicolumn{1}{r|}{}  & \multicolumn{1}{c}{W}  & \multicolumn{1}{c}{M}  & \multicolumn{1}{c|}{$\dmin$}  & \multicolumn{1}{c}{W}
   & \multicolumn{1}{c}{M}  & \multicolumn{1}{c}{$\dmin$}\\
    \hline
    Ord poly, true anom  & \multicolumn{1}{l|}{(OT)}  & 0  & 0  & 0  & $4\cdot10^{-6}$  & $4\cdot10^{-6}$  & $4\cdot10^{-6}$\\
    Ord poly, ecc anom  & \multicolumn{1}{l|}{(OE)}  & 0.0095  & 0.0221  & 0  & 0.0008  & 0.0633  & 0.0005\\
    Ord poly, ecc anom, shift  & \multicolumn{1}{l|}{(OES)}  & 0  & 0  & 0  & $1.1\cdot10^{-5}$  & $1.7\cdot10^{-5}$  & $4\cdot10^{-6}$ \\
    Trig poly, ecc anom & \multicolumn{1}{l|}{(TE)}  & 0  & 0.5732  & 0  & 0.0004  & 0.5234  & 0.0013\\
    Trig poly, ecc anom, Cheb & \multicolumn{1}{l|}{(TEC)}  & 0  & 0.0095  & 0  & 0.0003  & 0.0216  & 0.0003\\
    Trig poly, true anom & \multicolumn{1}{l|}{(TT)}  & 0  & 0.0253  & 0  & 0.0086  & 0.0408  & 0.0025\\
    Trig poly, true anom, shift & \multicolumn{1}{l|}{(TTS)}  & 0  & 0  & 0  &$1.3\cdot10^{-5}$  & 0.0006  & 0\\
    \hline
  \end{tabular}
  \caption{Percentages of detected errors with each method applied to
    the computation of all critical points between all NEAs and the
    Earth and between all pairs of NEAs.}
  \label{tab:err_data}
  \normalsize
\end{table}

Two additional ways to check in particular the computation of $\dmin$
are discussed below.

\subsection{Reliability test for $\dmin$}

Although all the presented methods allow us to find all the critical
points of $d^2$, we are particularly interested in the correct
computation of the minimum distance $\dmin$. For this reason, we
introduce two different tests to check whether the computed values of
$\dmin$ are reliable.

The first test is based on the results of \cite{gv13}, where the
authors found optimal upper bounds for $\dmin$ when one orbit is
circular.
Let us denote with ${\cal A}_1$ and ${\cal A}_2$ the two trajectories.
Assume that ${\cal A}_2$ is circular with orbital radius $r_2$, and
call $q_1,e_1,i_1,\omega_1$ the pericenter distance, eccentricity,
inclination and argument of pericenter of ${\cal A}_1$. Moreover, set
\begin{equation}
  {\cal C} = [0, 1] \times [0, \pi/2], \qquad
  {\cal D} = [0, q_{\rm max}] \times[0,\pi/2],
  \label{CDsets}
\end{equation}
where we used $q_{\rm max}=1.3$, which is the maximum perihelion
distance of near-Earth objects. Then, for each choice of
$(q_1,\omega_1)\in {\cal D}$ we have
\begin{equation}
  \max_{(e_1,i_1)\in{\cal C}}\dmin = \max\{r_2-q_1,\delta(q_1,\omega_1)\},
  \label{bounds_q_om}
\end{equation}
where $\delta(q_1,\omega_1)$ is the distance between ${\cal A}_1$ and
${\cal A}_2$ with $e_1=1, i_1=\pi/2$:
\begin{equation}
  \delta(q,\omega) = \sqrt{(\xi-r_2\sin\omega)^2 + \Bigl(\frac{\xi^2 - 4q^2}{4q} + r_2\cos\omega\Bigr)^2}\,,
  \label{deltaomega}
\end{equation}
with $\xi=\xi(q,\omega)$ the unique real solution of
\[
x^3 + 4q(q+\cos\omega)x - 8r_2q^2\sin\omega = 0.
\]
We compare this optimal bound with the maximum values of $\dmin$
computed with the method OE, for a grid of values in the
$(q_1,\omega_1)$ plane. The results are reported in Figure
\ref{maxdmin_qom}. Here we see that the maximum values of $\dmin$
obtained through our computation appear to lie on the grey surface
corresponding to the graph of $\displaystyle{\max_{{\cal
      C}}\dmin(q_1,\omega_1)}$ defined in \eqref{bounds_q_om}. This
test confirms the reliability of our computations.  Similar checks
were successful with all the methods of Table \ref{tab:err_data}.

\begin{figure}
  \centering
  \includegraphics[scale=.8]{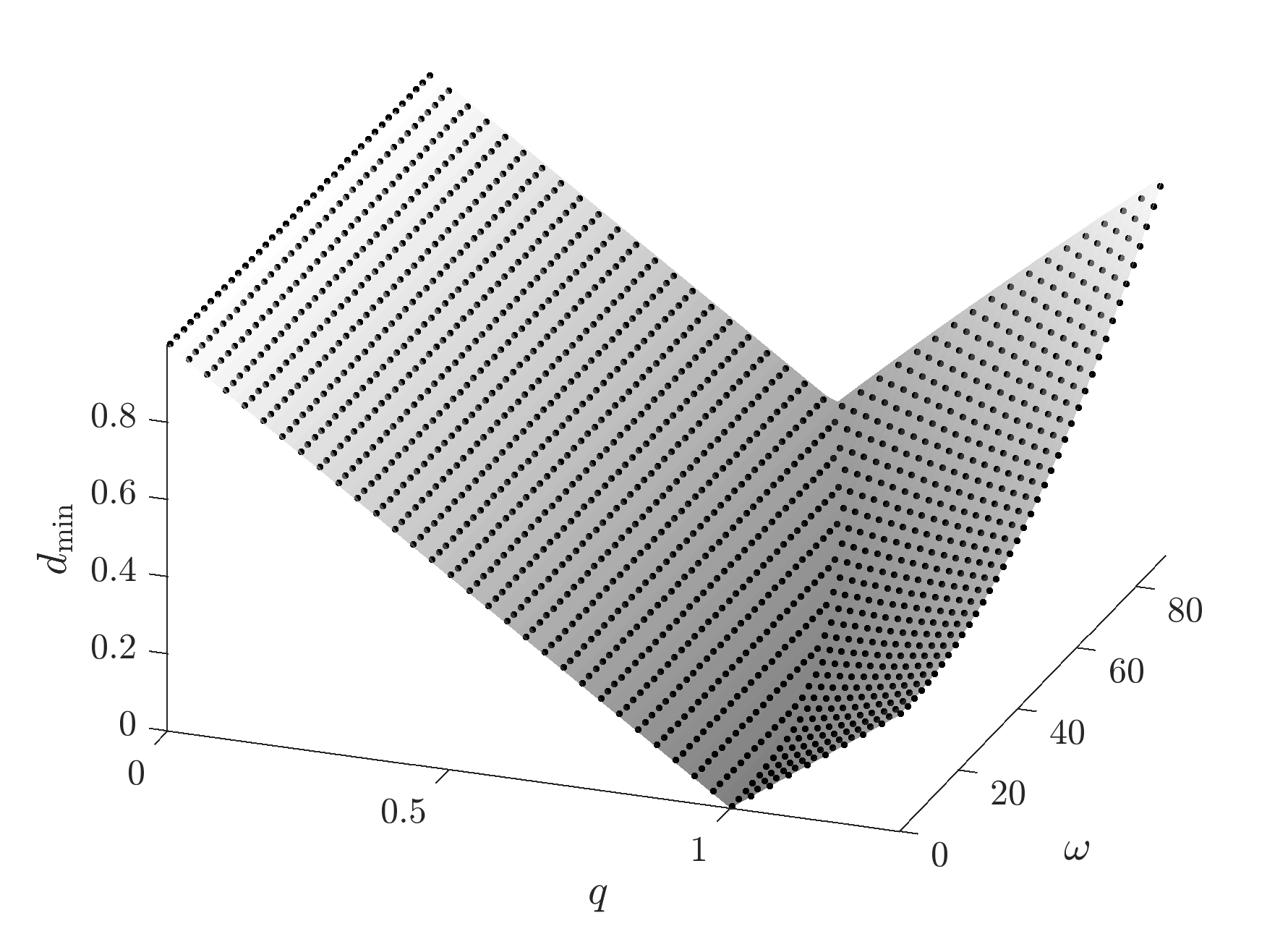}
  \caption{Comparison between the graph of
    $\max_{(e_1,i_1)}\dmin(q,\omega)$ defined in \eqref{bounds_q_om}
    and the maximum values of $\dmin$ computed with the algorithm of
    Section~\ref{s:ecc_reg} for a grid of values of $(q_1,\omega_1)$.}
  \label{maxdmin_qom}
\end{figure}

\medbreak To test our results also in case of two elliptic orbits, we
consider the following bound introduced in \cite{gn20} for the nodal
distance $\deltanod$ defined below. Let
\begin{align*}
  \rpiu &= \frac{q_1(1+e_1)}{1+e_1\cos\omegamutuno}, &
  \rmeno &= \frac{q_1(1+e_1)}{1-e_1\cos\omegamutuno},\\[0.5ex]
  \rpiup &= \frac{q_2(1+e_2)}{1+e_2\cos\omegamutdue}, &
  \rmenop &= \frac{q_2(1+e_2)}{1-e_2\cos\omegamutdue},
\end{align*}
where $q_2$, $e_2$ are the pericenter distance and eccentricity of
${\cal A}_2$, and $\omegamutuno$, $\omegamutdue$ are the mutual
arguments of pericenter (see \cite{gn20}).

We introduce the ascending and descending nodal distances
\[
\dnod^+ = \rpiup - \rpiu, \qquad \dnod^- = \rmenop - \rmeno.
\]
The (minimal) nodal distance $\deltanod$ is defined as
\begin{equation}
  \deltanod = \min\bigl\{|\dnod^+|, |\dnod^-|\bigr\}.
  \label{deltanod}
\end{equation}
Set
\[
{\cal C}' = [0, 1] \times [0, \pi].
\]
For each choice of $(q_1,\omega_1)\in{\cal D}$, defined as in
\eqref{CDsets}, we have
\begin{equation}
  \displaystyle \max_{(e_1,\omegamutdue)\in{\cal C}'} \deltanod = \max\bigl\{\uppintom, \uppextom, \upplinkom\bigr\},
  \label{maxdeltanodbound}
\end{equation}
where, denoting by $Q_2$ the apocenter distance of ${\cal A}_2$ and by
$p_2 = q_2(1+e_2)$ its conic parameter,
\[
\uppintom(q,\omega) = p_2 - q,
\]
\[
\uppextom(q,\omega) = \min\Bigl\{ \frac{2q}{1-\cos\omega} - \frac{p_2}{1-\hat{\xi}'_*},\ \frac{2q}{1+\cos\omega} - q_2\Bigr\},
\]
\[
\upplinkom(q,\omega) = \min\left\{ Q_2 - \frac{q(1+\hat{e}_*)}{1+\hat{e}_*\cos\omega},
\frac{2q}{1-\cos\omega}-q_2 \right\},
\]
with
\[
\hat{\xi}'_* = \min\{\xi_*',e_2\},\qquad
\xi'_*(q,\omega) = \frac{4q\cos\omega}{p_2\sin^2\omega+\sqrt{p_2^2\sin^4\omega+16q^2\cos^2\omega}},
\]
and
\begin{align*}
  & \hat{e}_* = \max\bigl\{0,\min\{e_*,1\}\bigr\},\\
  & e_*(q,\omega) = \frac{2(p_2-q(1-e_2^2))}{q(1-e_2^2)+\sqrt{q^2(1-e_2^2)^2+4p_2\cos^2\omega(p_2-q(1-e_2^2))}}.
\end{align*}

We compare the computed values of $\dmin$ with the bound
\eqref{maxdeltanodbound} on the maximum nodal distance. The results
are displayed in Figure \ref{deltanod_qom} where, for four different
values of $e_2$, the grey surface represents the bound of \cite{gn20},
while the black dots correspond to the maximum value of $\dmin$ for a
grid in the $(q_1,\omega_1)$ plane computed with the method OE.  Since
the value of $\deltanod$ is always greater than or equal to the value
of $\dmin$, for the test to be satisfied, we need all the black dots
to fall below or lie on the grey surface. From Figure
\ref{deltanod_qom} we see that this is indeed what happens.

Similar checks done with the methods OT, OE, OES, TT, TTS were
successful.

\begin{figure}
  \begin{subfigure}{.5\textwidth}
    \hspace{-3em}
    \includegraphics[width=1.1\textwidth]{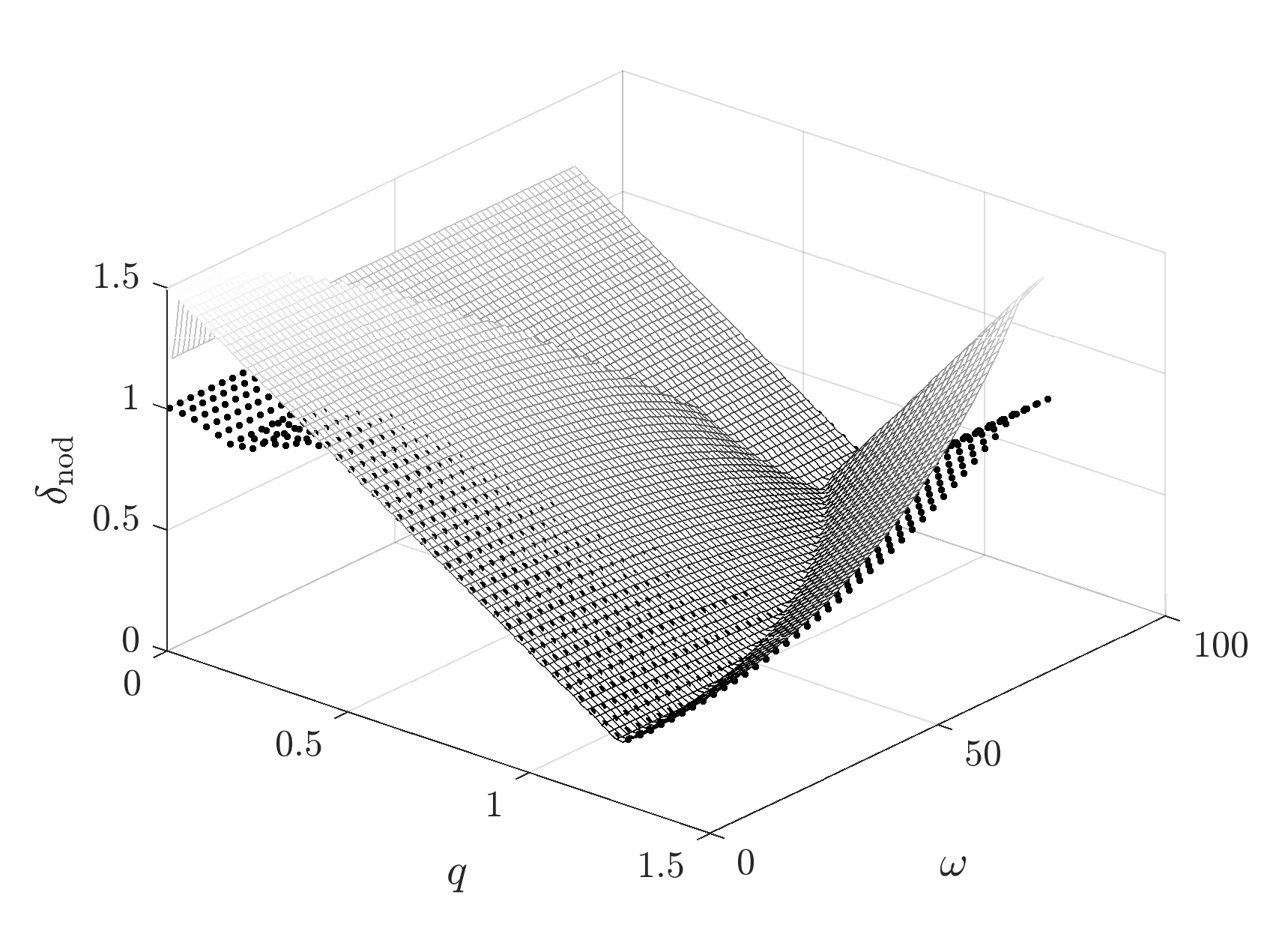}
  \end{subfigure}%
  \begin{subfigure}{.5\textwidth}
    \centering
    \includegraphics[width=1.1\textwidth]{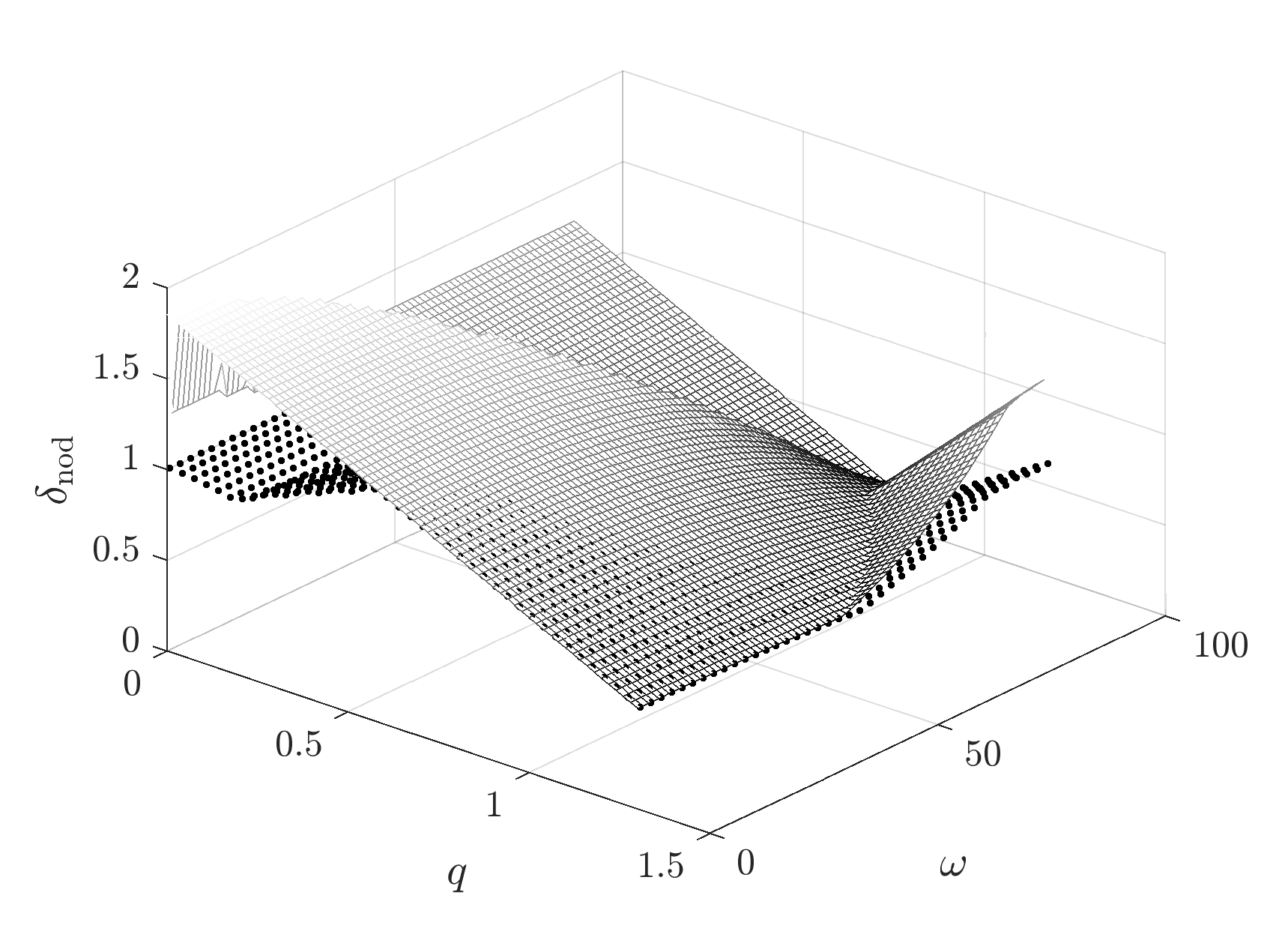}
  \end{subfigure}
  \begin{subfigure}{.5\textwidth}
    \hspace{-3em}
    \includegraphics[width=1.1\textwidth]{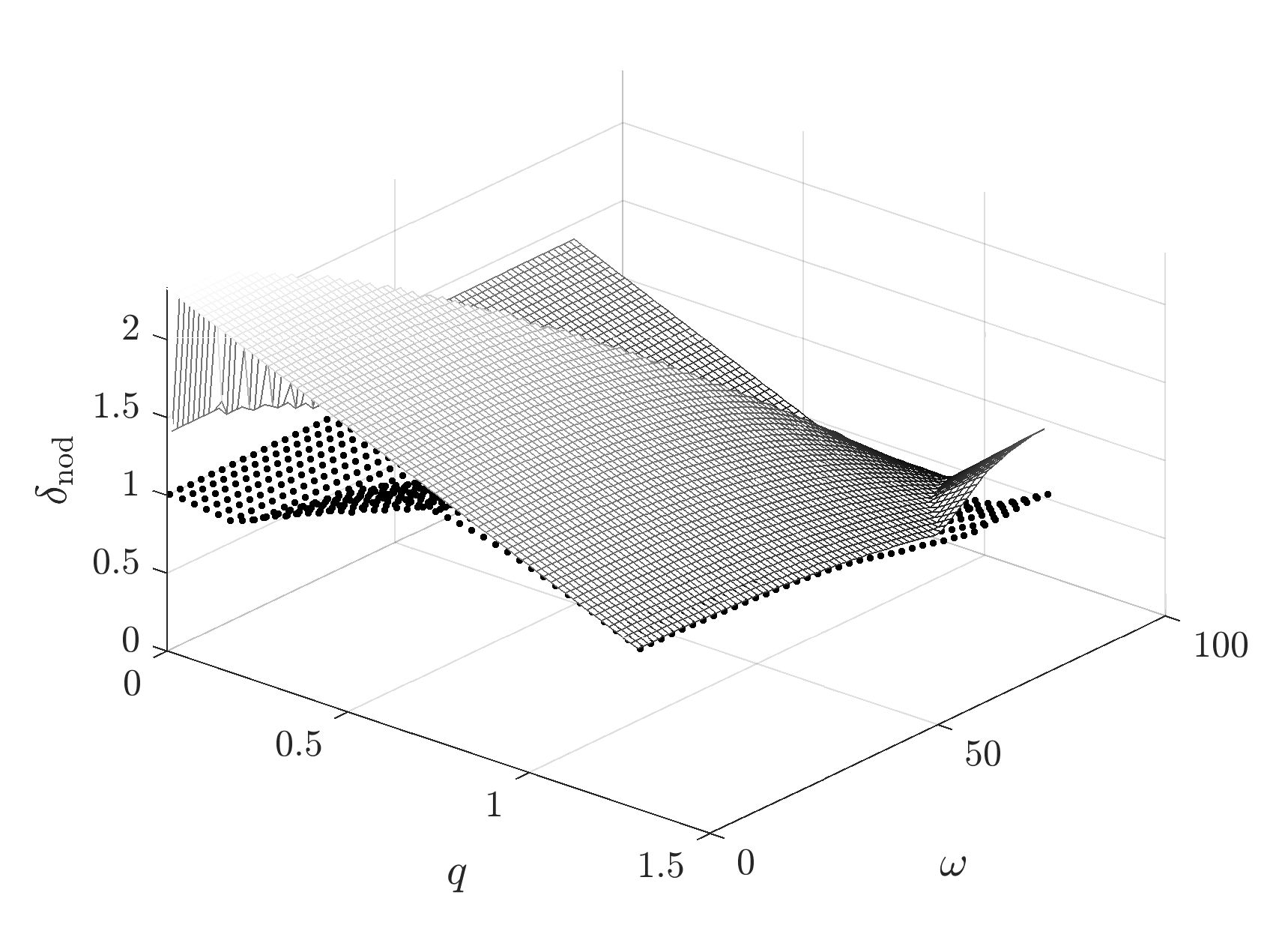}
  \end{subfigure}%
  \begin{subfigure}{.5\textwidth}
    \centering
    \includegraphics[width=1.1\textwidth]{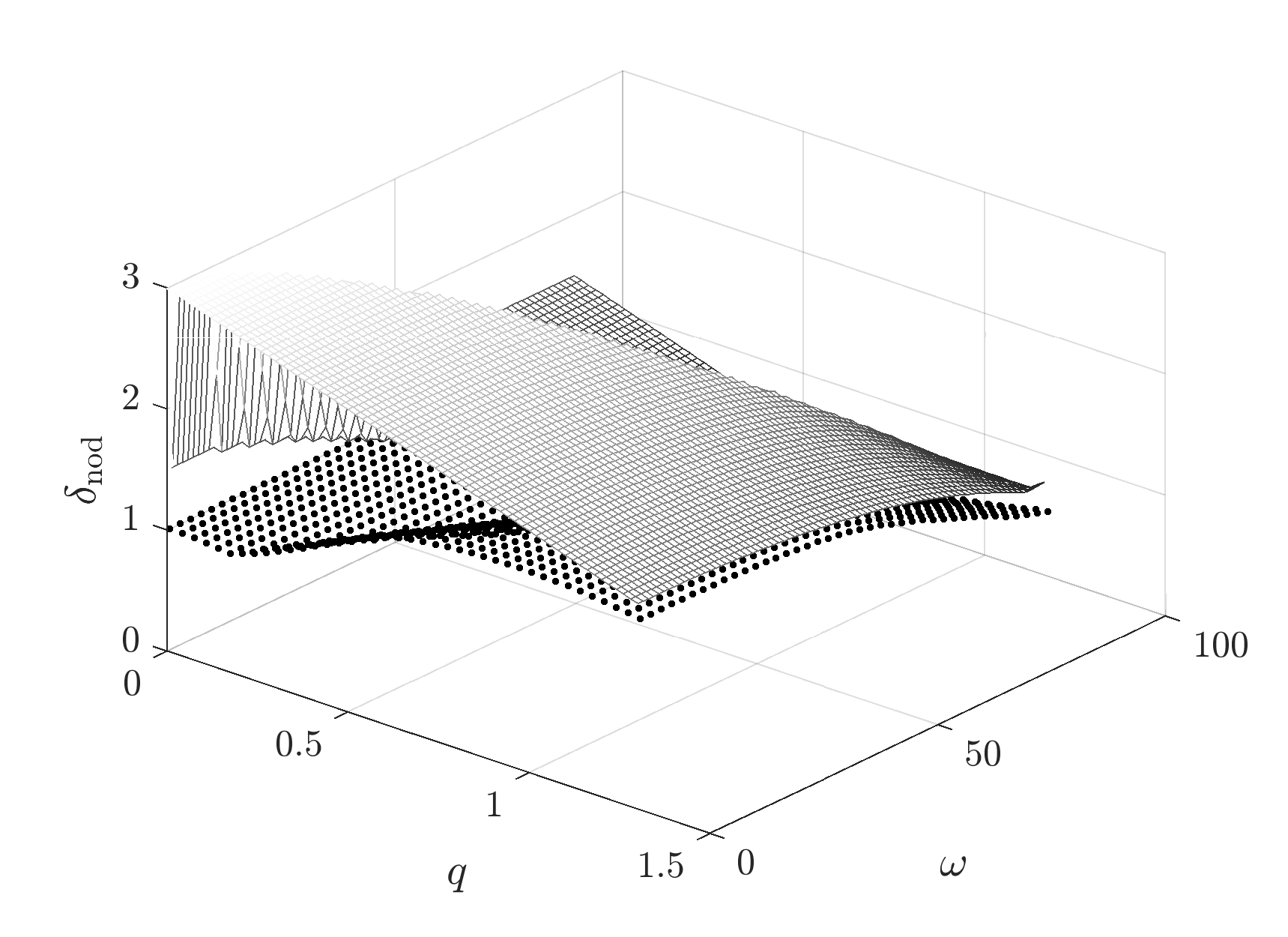}
  \end{subfigure}
  \caption{Comparison of the maximum MOID obtained with the method of
    Section \ref{s:ecc_reg} and the bound on the nodal distance found
    in \cite{gn20}. These plots were drawn using values of $e_2$ equal
    to 0.2 (top-left), 0.3 (top-right), 0.4 (bottom-left) and 0.5
    (bottom-right).}
  \label{deltanod_qom}
\end{figure}

\section{The planar case}
\label{s:planar}

Let us consider the case of two coplanar conics parametrized by the
true anomalies $f_1$, $f_2$.
Then, $(f_1, f_2)$ is a critical point of
$d^2$ iff $d^2(f_1, f_2)=0$ or the tangent vectors
\[
\boldsymbol{\tau}_1(f_1) = \frac{\partial {\cal X}_1}{\partial f_1},\qquad
\boldsymbol{\tau}_2(f_2) = \frac{\partial {\cal X}_2}{\partial f_2}
\]
to the first and second conic at ${\cal X}_1(f_1)$ and ${\cal
  X}_2(f_2)$, respectively, are parallel. If one trajectory, say the
second one, is circular, then the tangent vector $\boldsymbol{\tau}_2$
is orthogonal to the position vector ${\cal X}_2$ for any value of
$f_2$. Therefore, to find critical points that do not correspond to
trajectory intersections, it is enough to look for values of $f_1$
such that $\mathbf{r}_1\cdot\boldsymbol{\tau}_1=0$. By symmetry, we
can write
\[
{\cal X}_1 =
\begin{pmatrix}
  r_1\cos f_1\\
  r_1\sin f_1
\end{pmatrix},
\qquad\text{with}\quad r_1 = \frac{p_1}{1+e_1\cos f_1},
\]
that is we can assume $\omega_1 = 0$. Thus, up to a multiplicative
factor, we have
\[
\boldsymbol{\tau}_1 = 
\begin{pmatrix}
  -\sin f_1\\
  \cos f_1+e_1
\end{pmatrix},
\]
so that
\[
\mathbf{r}_1\cdot\boldsymbol{\tau}_1 = \frac{p_1e_1\sin f_1}{1+e_1\cos f_1} = 0
\]
is satisfied iff $f_1 = 0,\pi$. Therefore, in general, we have the
four critical points
\[
(\bar f_1,\bar f_2) = (0,0), (0,\pi), (\pi,0), (\pi,\pi).
\]
We may have at most two additional critical points that correspond to
trajectory intersections, see \cite[Sect. 7.1]{gt07}. In conclusion,
the maximum number of critical points with a circular and an elliptic
trajectory in the planar case is 6.

\medbreak We consider now the case of two ellipses. The position
vectors can be written as
\[
{\cal X}_1 =
\begin{pmatrix}
  r_1\cos f_1\\
  r_1\sin f_1
\end{pmatrix},
\qquad\text{with}\quad r_1 = \frac{p_1}{1+e_1\cos f_1}
\]
and
\[
{\cal X}_2 =
\begin{pmatrix}
  r_2\cos(f_2+\omega_2)\\
  r_2\sin(f_2+\omega_2)
\end{pmatrix},
\qquad\text{with}\quad r_2 = \frac{p_2}{1+e_2\cos f_2}.
\]
Up to a multiplicative factor, we have
\[
\bm{\tau}_1 =
\begin{pmatrix}
  -\sin f_1\\
  \cos f_1+e_1
\end{pmatrix},
\hskip 1cm
\bm{\tau}_2 =
\begin{pmatrix}
  -\sin(f_2+\omega_2)-e_2\sin\omega_2\\
  \hphantom{+}\cos(f_2+\omega_2)+e_2\cos\omega_2
\end{pmatrix}.
\]
The critical points that do not correspond to trajectory intersections
are given by the values of $f_1, f_2$ such that $\bm{\tau}_1,
\bm{\tau}_2$ are parallel and both orthogonal to ${\cal X}_2-{\cal
  X}_1$.  These two conditions lead to the system
\begin{empheq}[left=\empheqlbrace]{align}
  & (\cos f_1+e_1)[\sin(f_2+\omega_2)+e_2\sin\omega_2]-\sin f_1[\cos(f_2+\omega_2)+e_2\cos\omega_2] = 0,
  \label{eq:t1t2par}\\
  & (\cos f_1+e_1)[r_2\sin(f_2+\omega_2)-r_1\sin f_1]-\sin f_1[r_2\cos(f_2+\omega_2)-r_1\cos f_1] = 0.
  \label{eq:r21ort}
\end{empheq}
Multiplying \eqref{eq:t1t2par} by $r_2$ and subtracting
\eqref{eq:r21ort} we get
\begin{equation}
  p_1e_1(1+e_2\cos f_2)\sin f_1+p_2e_2(1+e_1\cos f_1)(\cos f_1\sin\omega_2-\sin f_1\cos\omega_2+e_1\sin\omega_2) = 0.
  \label{eqsubtract}
\end{equation}
Equations \eqref{eq:t1t2par} and \eqref{eqsubtract} can be written as
\begin{empheq}[left=\empheqlbrace]{align}
  & \alpha\cos f_2 + \beta\sin f_2 + \alpha e_2 = 0,
  \label{eq:planell2}\\
  & \mu\cos f_2 + \delta = 0,
  \label{eq:planell1}
\end{empheq}

\begin{table}[h!]
  \centering
  \begin{tabular}{ccccc}
    \hline
    \multicolumn{1}{c}{$q$}  & \multicolumn{1}{c}{$e$}  & \multicolumn{1}{c}{$i$}  & \multicolumn{1}{c}{$\Omega$}  & \multicolumn{1}{c}{$\omega$}\\
    \hline
    0.16582  & 0.84577  & 0  & 0  & 9.09466\\
    1        & 0.2      & 0  & 0  & 10\\
    \hline
  \end{tabular}
  \caption{Cometary orbital elements of a pair of coplanar elliptic
    orbits giving 10 critical points. Angles are in degrees.}
  \label{tab:10cpt}
\end{table}

\begin{figure}[h!]
  \centering
  \includegraphics[width=10cm]{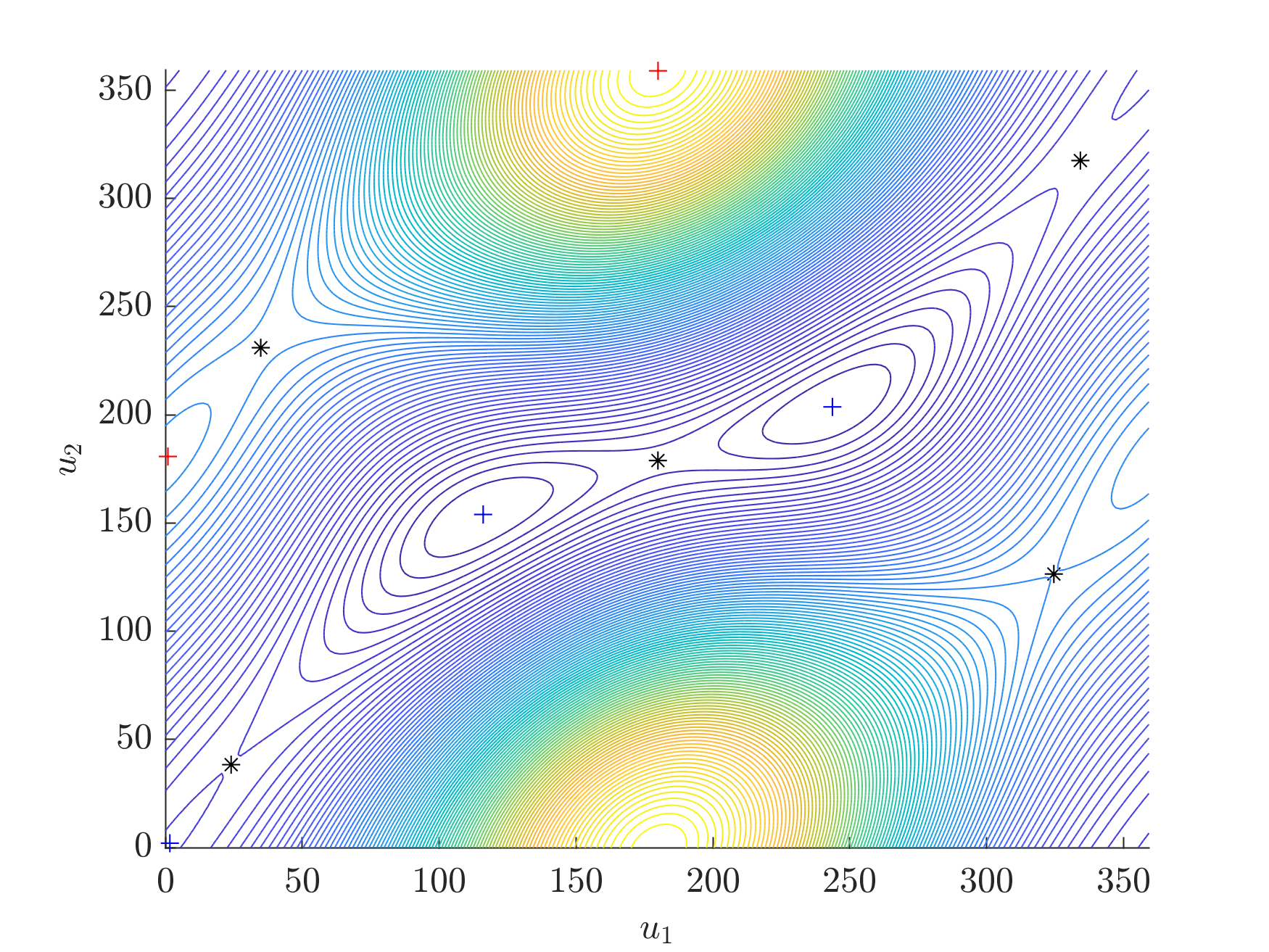}
  \caption{Level curves of the squared distance for the case reported
    in Table \ref{tab:10cpt}. The position of the critical points is
    highlighted: saddle points are represented by black asterisks,
    while the red and blue crosses correspond to maximum and minimum
    points, respectively.}
  \label{fig:lev_10cpt}
\end{figure}

\begin{table}[h!]
  \centering
  \begin{tabular}{d{3.7}|d{3.7}|c|c}
    \hline
    \multicolumn{1}{c}{$u_1$}  & \multicolumn{1}{c}{$u_2$}  & \multicolumn{1}{c}{$d$}  & \multicolumn{1}{c}{Type}\\
    \hline
    116.0625325  & 153.9899286  & 0.0000000  & MINIMUM\\
    243.6382848  & 203.6865581  & 0.0000000  & MINIMUM\\
    179.8948964  & 178.9198966  &  0.4845432  & SADDLE\\
    1.6247542  & 2.0946456  & 0.8341185  & MINIMUM\\
    24.0090191  & 38.3799855  & 0.8401907  & SADDLE\\
    334.2162041  & 317.5202237  & 0.8445898  & SADDLE\\
    324.5270438  & 126.4762243  & 1.6264123  & SADDLE\\
    34.8254033  & 231.0377067  & 1.6334795  & SADDLE\\
    0.9077692  & 180.7796090  & 1.6658557  & MAXIMUM\\
    179.9346562  & 358.9929507  & 2.9845260  & MAXIMUM\\
    \hline
  \end{tabular}
  \caption{Critical points and critical values for the pair of orbits
    displayed in Table \ref{tab:10cpt}.}
  \label{tab:cpt10}
\end{table}

where
\begin{align*}
  \alpha &= \sin(\omega_2-f_1)+e_1\sin\omega_2,\\
  \beta &= \cos(\omega_2-f_1)+e_1\cos\omega_2,\\  
  \mu &= p_1e_1e_2\sin f_1,\\
  \delta &= p_1e_1\sin f_1+p_2e_2(1+e_1\cos f_1)(\cos f_1\sin\omega_2-\sin f_1\cos\omega_2+e_1\sin\omega_2).
\end{align*}
From (\ref{eq:planell2}) we obtain
\[
\sin f_2 = -\frac{\alpha}{\beta}(\cos f_2+e_2),
\]
which is replaced into relation $\cos^2f_2+\sin^2f_2-1=0$ to give
\begin{equation}
  (\alpha^2+\beta^2)\cos^2f_2+2\alpha^2e_2\cos f_2+e_2^2\alpha^2-\beta^2 = 0.
  \label{eq:planell3}
\end{equation}
Moreover, after replacing in (\ref{eq:planell3})
\[
\cos f_2 = -\frac{\delta}{\mu},
\]
which follows from (\ref{eq:planell1}), we obtain
\begin{equation}
  (\alpha^2+\beta^2)\delta^2-2e_2\alpha^2\delta\mu+\mu^2(e_2^2\alpha^2-\beta^2) = 0.
  \label{eq:fineq}
\end{equation}
Since
\[
\alpha^2+\beta^2=1+e_1^2+2e_1\cos f_1,
\]
the trigonometric polynomial in (\ref{eq:fineq}) has, in general,
degree 5 in $\cos f_1$, $\sin f_1$. Therefore, we can not have more
than 10 critical points which do not correspond to trajectory
intersections. Then, the maximum number of critical points of $d^2$
(including intersections) for two elliptic orbits in the planar case
is at most 12. However, we remark that this bound has never been
reached in our numerical tests, where we got at most 10 critical
points that we think is the maximum number. This conjecture adds a new
question to Problem 8 in \cite{Albouy2012}.

In Table \ref{tab:10cpt} we write a set of orbital elements giving 10
critical points. We draw the level curves of $d^2$ in Figure
\ref{fig:lev_10cpt}, as function of the eccentric anomalies $u_1,
u_2$, where the position of each critical point is highlighted: we use
an asterisk for saddle points, and crosses for local extrema.
Finally, the critical points, the corresponding values of $d$ and
their type (minimum, maximum, saddle) are displayed in
Table~\ref{tab:cpt10}.

\section{Conclusions}
\label{s:conclusions}

In this work we investigate different approaches for the computation
of the critical points of the squared distance function $d^2$, with
particular care for the minimum values.
We focus on the case of bounded trajectories. Two algebraic approaches
are used: the first employs ordinary polynomials, the second
trigonometric polynomials. In both cases we detail all the steps to
reduce the problem to the computation of the roots of a univariate
polynomial of minimal degree (that is 16, in the general case). The
different methods are compared through numerical tests using the
orbits of all the known near-Earth asteroids. We also perform some
reliability tests of the results, which make use of known optimal
bounds on the orbit distance. Finally, we improve the theoretical
bound on the number of critical points in the planar case, and refine
the related conjecture.

\section{Acknowledgments}

We wish to thank Leonardo Robol for his useful comments and
suggestions.  The authors have been partially supported through the
H2020 MSCA ETN Stardust-Reloaded, Grant Agreement n. 813644. The
authors also acknowledge the project MIUR-PRIN 20178CJA2B ``New
frontiers of Celestial Mechanics: theory and applications'' and the
GNFM-INdAM (Gruppo Nazionale per la Fisica Matematica).

\appendix
\section{Coefficients of shifted polynomials for method with ordinary polynomials and eccentric anomalies}
\label{a:shift_ord}

The coefficients of system \eqref{polysys_shift} are
\begin{align*}
    \tilde\alpha &= \bigl( 2(A_1-A_3)\cos s_1\sin s_1 - A_{10}\sin s_1 - A_{11}\cos s_1 + A_{12}\sin s_1 + A_8\cos s_1\bigr)z^4\\
    &\quad\, + \bigl(-8(A_1-A_3)\cos^2s_1 + 2A_{10}\cos s_1 - 2A_{11}\sin s_1\bigr.\\
    &\quad\, \bigl.{} - 2 A_{12}\cos s_1 + 2 A_8\sin s_1 + 4(A_1-A_3) \bigr) z^3 - 12(A_1-A_3)\sin s_1 \cos s_1 z^2\\
    &\quad\, + \bigl(8(A_1-A_3)\cos^2s_1 + 2A_{10}\cos s_1 - 2A_{11}\sin s_1 - 2A_{12}\cos s_1 \bigr.\\
    &\quad\, \bigl.{} + 2A_8\sin s_1 - 4(A_1-A_3)\bigr)z + 2(A_1-A_3)\cos s_1\sin s_1\\
    &\quad\, + A_{10}\sin s_1 + A_{11}\cos s_1 - A_{12} \sin s_1 - A_8\cos s_1 ,\\[1ex]
    \tilde\beta &= \left(-2A_7\cos s_1 + 2A_9\sin s_1\right) z^4 + \left(-4A_7\sin s_1 - 4A_9\cos s_1\right)z^3\\
    &\quad\, + \left(-4 A_7\sin s_1 - 4 A_9\cos s_1 \right)z + 2A_7\cos s_1 - 2A_9\sin s_1,\\[1ex]
    \tilde\gamma &= \bigl(2(A_1-A_3)\cos s_1 \sin s_1 + A_{10}\sin s_1 - A_{11}\cos s_1 + A_{12}\sin s_1 - A_8\cos s_1\bigr)z^4 \\
    &\quad\, + \bigl(-8 (A_1-A_3)\cos^2 s_1 - 2A_{10}\cos s_1 - 2A_{11}\sin s_1 \bigr.\\
    &\quad\,\bigl.{} - 2A_{12}\cos s_1 - 2A_8\sin s_1 + 4(A_1-A_3)\bigr)z^3 - 12(A_1-A_3)\sin s_1\cos s_1 z^2\\
    &\quad\, + \bigl(8(A_1-A_3)\cos^2 s_1 - 2 A_{10}\cos s_1 - 2 A_{11}\sin s_1 - 2 A_{12}\cos s_1 \bigr.\\
    &\quad\,\bigl. {} -2A_8\sin s_1 - 4(A_1-A_3)\bigr)z + 2(A_1-A_3)\cos s_1 \sin s_1\\
    &\quad\, - A_{10}\sin s_1 + A_{11}\cos s_1 - A_{12}\sin s_1 + A_8\cos s_1 ,\\[1ex]
    \tilde A &= \left(A_7\sin s_1 + A_9\cos s_1 - A_{13} \right)z^2 + \left(-2 A_7\cos s_1 + 2 A_9\sin s_1\right)z\\
    &\quad\, - A_7\sin s_1 - A_9\cos s_1 - A_{13} ,\\[1ex]
    \tilde B &= \bigl(2A_{10}\cos s_1 + 2A_8\sin s_1 - 2A_{14} - 4(A_4-A_6)\bigr) z^2 + \left(4A_{10}\sin s_1 - 4A_8\cos s_1\right)z\\
    &\quad\, - 2A_{10}\cos s_1 - 2A_8\sin s_1 - 2 A_{14} - 4(A_4-A_6) ,\\[1ex]
    \tilde D &= \bigl(2A_{10}\cos s_1 + 2A_8\sin s_1 - 2A_{14} + 4(A_4-A_6)\bigr) z^2 + \left(4A_{10}\sin s_1 - 4 A_8\cos s_1\right)z\\
    &\quad\, - 2 A_{10}\cos s_1 - 2 A_8\sin s_1 - 2 A_{14} + 4(A_4-A_6) .
\end{align*}

\section{Factorization of $\mathscr{H}(f_2)$}
\label{a:true_anom}
Let
\[
\mathscr{T} =
\left[
  \begin{array}{cccc}
    \alpha^2+\beta^2  & 0  & \beta\kappa-\alpha\lambda  & 0\cr
    2\alpha\gamma  & \alpha^2+\beta^2  & \beta\mu-\lambda\gamma-\alpha\nu  & \beta\kappa-\alpha\lambda\cr
    \gamma^2-\beta^2  & 2\alpha\gamma  & \beta\kappa-\gamma\nu  & \beta\mu-\lambda\gamma-\alpha\nu\cr
    0  & \gamma^2-\beta^2  & 0  & \beta\kappa-\gamma\nu\cr
  \end{array}
  \right]
\]
and define
\[
\mathscr{H}(f_2) = \det\mathscr{T}.
\]

\begin{proposition}
  We can extract the factor $\beta^2(1+e_2\cos f_2)^2$ from
  $\det\mathscr{T}$.
\end{proposition}

\begin{proof}
  We first prove that we can extract the factor $\beta^2$.

  \noindent Noting that
  \begin{equation}
    \kappa = \mu\eta,\qquad \eta=\frac{e_1}{1+e_1^2},
    \label{eq:tau}
  \end{equation}
  we consider
  \[
  \mathscr{T} =
  \left[
    \begin{array}{cccc}
      \alpha^2+\beta^2  & 0  & \beta\mu\eta-\alpha\lambda  & 0\cr
      2\alpha\gamma  & \alpha^2+\beta^2  & \beta\mu-\lambda\gamma-\alpha\nu  & \beta\mu\eta-\alpha\lambda\cr
      \gamma^2-\beta^2  & 2\alpha\gamma  & \beta\mu\eta-\gamma\nu  & \beta\mu-\lambda\gamma-\alpha\nu\cr
      0  & \gamma^2-\beta^2  & 0  & \beta\mu\eta-\gamma\nu\cr
    \end{array}
    \right].
  \]
  We can write $\det\mathscr{T}$ as a sum of terms where the only one
  that is independent on $\beta$ is
  \[
  \left|
  \begin{array}{cccc}
    \alpha^2  & 0  & -\alpha\lambda  & 0\cr
    2\alpha\gamma  & \alpha^2  & -\lambda\gamma-\alpha\nu  & -\alpha\lambda\cr
    \gamma^2  & 2\alpha\gamma  & -\gamma\nu  & -\lambda\gamma-\alpha\nu\cr
    0  & \gamma^2  & 0  & -\gamma\nu\cr
  \end{array}
  \right|,
  \]
  which is equal to $0$, as previously proved at the beginning of
  Proposition \ref{prop:ea_factor}. The terms that are linearly
  dependent on $\beta$ are given by
  \[
  \left|
  \begin{array}{cccc}
    \alpha^2  & 0  & -\alpha\lambda  & 0\cr
    2\alpha\gamma  & \alpha^2  & -\lambda\gamma-\alpha\nu  & \beta\mu\eta\cr
    \gamma^2  & 2\alpha\gamma  & -\gamma\nu  & \beta\mu\cr
    0  & \gamma^2  & 0  & \beta\mu\eta\cr
  \end{array}
  \right|,
  \hskip 1cm
  \left|
  \begin{array}{cccc}
    \alpha^2  & 0  & \beta\mu\eta  & 0\cr
    2\alpha\gamma  & \alpha^2  & \beta\mu  & -\alpha\lambda\cr
    \gamma^2  & 2\alpha\gamma  & \beta\mu\eta  & -\lambda\gamma-\alpha\nu\cr
    0  & \gamma^2  & 0  & -\gamma\nu\cr
  \end{array}
  \right|,
  \]
  and their sum is equal to $0$, because the two determinants are
  opposite. Therefore, $\mathscr{H}(f_2)$ is made by terms of degree
  higher than 1 in $\beta$. Thus, we can write
  \[
  \mathscr{H}(f_2) = \mathfrak{D}_1 + \mathfrak{D}_2 + \mathfrak{D}_3 + \mathfrak{D}_4 + \mathfrak{D}_5,
  \]
  where
  \[
  \mathfrak{D}_1 =
  \left|
  \begin{array}{cccc}
    \beta^2  & 0  & \beta\mu\eta  & 0\cr
    0  & \beta^2  & \beta\mu  & \beta\mu\eta\cr
    -\beta^2  & 0  & \beta\mu\eta  & \beta\mu\cr
    0  & -\beta^2  & 0  & \beta\mu\eta\cr
  \end{array}
  \right|,
  \]
  \[
  \mathfrak{D}_2 =
  \left|
  \begin{array}{cccc}
    \beta^2  & 0  & \beta\mu\eta  & 0\cr
    0  & \beta^2  & \beta\mu  & -\alpha\lambda\cr
    -\beta^2  & 0  & \beta\mu\eta  & -\lambda\gamma-\alpha\nu\cr
    0  & -\beta^2 & 0  & -\gamma\nu\cr
  \end{array}
  \right|
  +
  \left|
  \begin{array}{cccc}
    \beta^2  & 0  & -\alpha\lambda  & 0\cr
    0  & \beta^2  & -\lambda\gamma-\alpha\nu  & \beta\mu\eta\cr
    -\beta^2  & 0  & -\gamma\nu  & \beta\mu\cr
    0  & -\beta^2  & 0  & \beta\mu\eta\cr
  \end{array} 
  \right|,
  \]
  \begin{align*}
    \mathfrak{D}_3 &=
    \left|
    \begin{array}{cccc}
      \beta^2  & 0  & -\alpha\lambda  & 0\cr
      0  & \beta^2  & -\lambda\gamma-\alpha\nu  & -\alpha\lambda\cr
      -\beta^2  & 0  & -\gamma\nu  & -\lambda\gamma-\alpha\nu\cr
      0  & -\beta^2  & 0  & -\gamma\nu\cr
    \end{array}
    \right|
    +
    \left|
    \begin{array}{cccc}
      \alpha^2  & 0  & \beta\mu\eta  & 0\cr
      2\alpha\gamma  & \beta^2  & \beta\mu  & \beta\mu\eta\cr
      \gamma^2  & 0  & \beta\mu\eta  & \beta\mu\cr
      0  & -\beta^2  & 0  &\beta\mu\eta\cr
    \end{array}
    \right|
    \\
    & \quad\,
    +\left|
    \begin{array}{cccc}
      \beta^2  & 0  & \beta\mu\eta  & 0\cr
      0  & \alpha^2  & \beta\mu  & \beta\mu\eta\cr
      -\beta^2  & 2\alpha\gamma  & \beta\mu\eta  & \beta\mu\cr
      0  & \gamma^2   & 0  & \beta\mu\eta\cr
    \end{array}
    \right|,
  \end{align*}
  \begin{align*}
    \mathfrak{D}_4 &=
    \left|
    \begin{array}{cccc}
      \alpha^2  & 0  & \beta\mu\eta  & 0\cr
      2\alpha\gamma  & \beta^2  & \beta\mu  & -\alpha\lambda\cr
      \gamma^2  & 0  & \beta\mu\eta  & -\lambda\gamma-\alpha\nu\cr
      0  & -\beta^2  & 0  & -\gamma\nu\cr
    \end{array}
    \right|
    +
    \left|
    \begin{array}{cccc}
      \alpha^2  & 0  & -\alpha\lambda  & 0\cr
      2\alpha\gamma  & \beta^2  & -\lambda\gamma-\alpha\nu  &\beta\mu\eta\cr
      \gamma^2  & 0  & -\gamma\nu  & \beta\mu\cr
      0  & -\beta^2  & 0  & \beta\mu\eta\cr
    \end{array}
    \right|
    \\
    & \quad\,
    +\left|
    \begin{array}{cccc}
      \beta^2  & 0  & \beta\mu\eta  & 0\cr
      0  & \alpha^2  & \beta\mu  & -\alpha\lambda\cr
      -\beta^2  & 2\alpha\gamma  & \beta\mu\eta  & -\lambda\gamma-\alpha\nu\cr
      0  & \gamma^2  & 0  & -\gamma\nu\cr
    \end{array}
    \right|
    +
    \left|
    \begin{array}{cccc}
      \beta^2  & 0  & -\alpha\lambda  & 0\cr
      0  & \alpha^2  & -\lambda\gamma-\alpha\nu  & \beta\mu\eta\cr
      -\beta^2  & 2\alpha\gamma  & -\gamma\nu  & \beta\mu\cr
      0  & \gamma^2  & 0  & \beta\mu\eta\cr
    \end{array}
    \right|,
  \end{align*}
  \begin{align*}
    \mathfrak{D}_5 &=
    \left|
    \begin{array}{cccc}
      \alpha^2  & 0  & -\alpha\lambda  & 0\cr
      2\alpha\gamma  & \beta^2  & -\lambda\gamma-\alpha\nu  & -\alpha\lambda\cr
      \gamma^2  & 0  & -\gamma\nu  & -\lambda\gamma-\alpha\nu\cr
      0  & -\beta^2  & 0  & -\gamma\nu\cr
    \end{array}
    \right|
    +
    \left|
    \begin{array}{cccc}
      \beta^2  & 0  & -\alpha\lambda  & 0\cr
      0  & \alpha^2  & -\lambda\gamma-\alpha\nu  & -\alpha\lambda\cr
      -\beta^2  & 2\alpha\gamma  & -\gamma\nu  & -\lambda\gamma-\alpha\nu\cr
      0  & \gamma^2  & 0  & -\gamma\nu\cr
    \end{array}
    \right|
    \\
    & \quad\, + 
    \left|
    \begin{array}{cccc}
      \alpha^2  & 0  & \beta\mu\eta-\alpha\lambda  & 0\cr
      2\alpha\gamma  & \alpha^2  & \beta\mu-\lambda\gamma-\alpha\nu  & \beta\mu\eta-\alpha\lambda\cr
      \gamma^2  & 2\alpha\gamma  & \beta\mu\eta-\gamma\nu  & \beta\mu-\lambda\gamma-\alpha\nu\cr
      0  & \gamma^2  & 0  & \beta\mu\eta-\gamma\nu\cr
    \end{array}
    \right|.
  \end{align*}
  We have
  \begin{align*}
    \mathfrak{D}_1 &= \beta^6\mu^2(4\eta^2-1),\\[1ex]
    \mathfrak{D}_2 &= 2\beta^5\mu\left[\lambda\gamma+\alpha\nu-2\eta(\alpha\lambda+\gamma\nu)\right],\\[1ex]
    \mathfrak{D}_3 &= \beta^4(\alpha^2-\gamma^2)\left[\lambda^2-\nu^2+\mu^2(4\eta^2-1)\right],\\[1ex]
    \mathfrak{D}_4 &= -2\beta^3\mu\left[\gamma^3\lambda-\alpha^3\nu+\eta(\alpha^3\lambda-3\alpha\gamma^2\lambda+
      3\alpha^2\gamma\nu-\gamma^3\nu)\right],\\[1ex]
    \mathfrak{D}_5 &= \beta^2\mu^2\left[\alpha\gamma-\eta(\alpha^2+\gamma^2)\right]^2-\beta^2(\alpha^2-\gamma^2)
    (\lambda\gamma-\alpha\nu)^2.
  \end{align*}

  \begin{remark}
    $\mathscr{H}(f_2)/\beta^2$ is a trigonometric polynomial of degree
    10 in $(\cos f_2$, $\sin f_2)$.
  \end{remark}

  Then, we show that $\xi^2 = (1+e_2\cos f_2)^2$ is a factor of
  $\mathscr{H}(f_2)/\beta^2$. For this purpose, using the definitions
  of $\alpha$, $\beta$, $\nu$, $\lambda$, $\tilde{\alpha}$,
  $\tilde{\beta}$, $\tilde{\lambda}$ given in
  Section~\ref{s:true_anom}, we write
  \begin{align}
    \alpha &= \xi\tilde{\alpha}+e_1\gamma,\label{eq:alpha}\\
    \beta &= \xi\tilde{\beta},\nonumber\\
    \lambda &= e_1\tilde{\lambda},\label{eq:lambda}\\
    \nu &= p_1e_1\xi +\tilde{\lambda}\label{eq:nu}.
  \end{align}
  The factor $\beta^2$ can be extracted from
  $(\mathfrak{D}_1+\mathfrak{D}_2+\mathfrak{D}_3)/\beta^2$, therefore
  also $\xi^2$ is a factor of this polynomial. Consider now
  $\mathfrak{D}_4/\beta^2$ and write it as
  \[
  \mathfrak{D}_4/\beta^2 = 6\mu\beta\alpha\gamma(\gamma\lambda-\alpha\nu)-2\mu\beta
  \left[\lambda(\gamma^3+\eta\alpha^3)-\nu(\alpha^3+\eta\gamma^3)\right].
  \]
  Noting that
  \begin{equation}
    \gamma\lambda-\alpha\nu = -\xi(\nu\tilde{\alpha}+p_1e_1^2\gamma)
    \label{eq:aux}
  \end{equation}
  and
  \[
  \gamma^3\left[\lambda(1+\eta e_1^3)-\nu(e_1^3+\eta)\right] = -\gamma^3(1+\eta e_1^3)\xi p_1e_1^2,
  \]
  where we used the expressions of $\alpha$, $\lambda$, $\nu$ in
  \eqref{eq:alpha}, \eqref{eq:lambda}, \eqref{eq:nu} and the
  definition of $\eta$ in \eqref{eq:tau}, we prove that $\xi^2$
  factors $\mathfrak{D}_4/\beta^2$.

  Finally, using relation \eqref{eq:aux} and
  \[
  \alpha\gamma-\eta(\alpha^2+\gamma^2) =
  \xi\tilde{\alpha}\left[\gamma(1-2\eta e_1)-\xi\eta\tilde{\alpha}\right],
  \]
  we show that also $\mathfrak{D}_5/\beta^2$ contains the factor $\xi^2$.
\end{proof}

The trigonometric polynomial
\begin{align*}
  h(f_2) &= \mathscr{H}(f_2)/(\beta^2\xi^2)\\
  &= \tilde{\beta}^4\xi^2\mu^2(4\eta^2-1)
  +2\tilde{\beta}^3\xi\mu\left[\lambda\gamma+\alpha\nu-2\eta(\alpha\lambda+\gamma\nu)\right]\\
  &\,\quad +\tilde{\beta}^2(\alpha^2-\gamma^2)\left[\lambda^2-\nu^2+\mu^2(4\eta^2-1)\right]\\
  &\,\quad -2\mu\tilde{\beta}\tilde{\alpha}^2\xi^2\left[\tilde{\alpha}(\lambda\eta-\nu)-3\eta e_1^3 p_1\gamma\right]
  + \mu^2\tilde{\alpha}^2\left[\gamma(1-2\eta e_1)-\xi\eta\tilde{\alpha}\right]^2\\
  &\,\quad -2\mu\tilde{\beta}p_1e_1(1-e_1^2)\eta\gamma^2\left[3e_1\tilde{\alpha}\xi-\gamma(1-e_1^2)\right]
  -(\alpha^2-\gamma^2)(\nu\tilde{\alpha}+p_1e_1^2\gamma)^2
\end{align*}
is of degree 8 in $(\cos f_2$, $\sin f_2)$.

\section{Angular shift for trigonometric polynomials}
\label{a:shift}
Let
\[
p(x,y) = \sum_{(i,j)\in {\cal K}}p_{i,j} x^i y^j,
\]
with $x = \cos u, y = \sin u$ and ${\cal K}\subset\N\times\N$
(non-negative 2-index integers) be a trigonometric polynomial.  We
wish to write $p(x,y)$ in terms of the variables $(z,w)$, where
\[
z = \cos v,\quad w = \sin v, \qquad v = u - \alpha, \quad \alpha\in \R.
\]
Writing $c_\alpha$, $s_\alpha$ for $\cos\alpha$, $\sin\alpha$,
respectively, we have
\[
x = c_\alpha z - s_\alpha w, \qquad y = s_\alpha z + c_\alpha w,
\]
so that we obtain
\[
\begin{split}
  q(z,w) &= p(c_\alpha z - s_\alpha w,s_\alpha z + c_\alpha w)\\
  & = \sum_{(i,j)\in {\cal K}}p_{i,j}\left[\sum_{h=0}^i\biggl(\begin{array}{c}i\cr h\end{array}\biggr)(c_\alpha z)^h(-s_\alpha w)^{i-h}\right]
  \left[\sum_{k=0}^j\biggl(\begin{array}{c}j\cr k\end{array}\biggr)(s_\alpha z)^k(c_\alpha w)^{j-k}\right]\\
  & = \sum_{(i,j)\in {\cal K}}p_{i,j}\sum_{\ell=0}^{i+j}\sum_{h+k=\ell}\biggl(\begin{array}{c}i\cr h\end{array}\biggr)
  \biggl(\begin{array}{c}j\cr k\end{array}\biggr)(-1)^{i-h} (c_\alpha)^{h+j-k}(s_\alpha)^{i-h+k} z^{h+k} w^{i-h+j-k}\\
  & = \sum_{(i,j)\in {\cal K}}p_{i,j}\sum_{\ell=0}^{i+j}\sum_{h=\max\{\ell-j,0\}}^{\min\{\ell,i\}}\biggl(\begin{array}{c}i\cr h\end{array}\biggr)
  \biggl(\begin{array}{c}j\cr \ell-h\end{array}\biggr)(-1)^{i-h} (c_\alpha)^{2h+j-\ell}(s_\alpha)^{i-2h+\ell} z^\ell w^{i+j-\ell}. 
\end{split}
\]
If
\[
{\cal K} = \{0,1,\ldots,m\}\times \{0,1,\ldots,n\},
\]
introducing the coefficients
\[
C_{i,j,\ell} = \sum_{h=\max\{\ell-j,0\}}^{\min\{\ell,i\}}\biggl(\begin{array}{c}i\cr h\end{array}\biggr)\biggl(\begin{array}{c}j\cr \ell - h\end{array}\biggr)
(-1)^{i-h} (c_\alpha)^{2h+j-\ell}(s_\alpha)^{i-2h+\ell},
\]
we can write
\[
\begin{split}
q(z,w) &= \sum_{i=0}^m\sum_{j=0}^n p_{i,j}\sum_{\ell=0}^{i+j}C_{i,j,\ell} z^\ell w^{i+j-\ell} = \sum_{r = 0}^{m+n} \sum_{i=\max\{r-n,0\}}^{\min\{r,m\}} p_{i,r-i}
\sum_{\ell=0}^{\cancel{i}+r\cancel{-i}}C_{i,r-i,\ell} z^\ell w^{\cancel{i}+r\cancel{-i}-\ell}\\
&= \sum_{r = 0}^{m+n}\sum_{\ell=0}^r q_{\ell,r-\ell} z^\ell w^{r-\ell},
\end{split}
\]
where
\[
q_{\ell,r-\ell} = \sum_{i=\max\{r-n,0\}}^{\min\{r,m\}} p_{i,r-i}\, C_{i,r-i,\ell}.
\]

\bibliography{mybib}{}
\bibliographystyle{plain}

\end{document}